\documentclass[12pt]{article}
\usepackage[utf8]{inputenc}
\usepackage{amsmath, amsthm, mathtools, amsfonts, graphicx, hyperref,setspace, xparse, algorithm, algpseudocode, tikz, longtable, times, scalerel, mathabx, float, makecell}
\usepackage[margin=1in]{geometry}

\theoremstyle{plain}
\newtheorem{theorem}{Theorem}[section]
\newtheorem{lemma}[theorem]{Lemma}
\newtheorem{prop}[theorem]{Proposition}

\theoremstyle{definition}
\newtheorem{definition}{Definition}[section]
\newtheorem{example}{Example}[section]
\newtheorem{numexmp}{Real Data Example}[section]
\newtheorem{numexmp1}{Numerical Example}[section]
\theoremstyle{remark}
\newtheorem{remark}{Remark}

\title{Estimating MCMC convergence rates using common random number simulations}
\author{Sabrina Sixta$^1$\thanks{email: sabrina.sixta@mail.utoronto.ca}, Jeffrey S. Rosenthal$^1$\thanks{email: jeff@math.toronto.edu}, and Austin Brown$^2$\thanks{email: austinbrown@tamu.edu}\\
$^1$ Department of Statistical Sciences, University of Toronto, Toronto, ON, 
Canada\\
$^2$ Department of Statistics, Texas A\&M University, College Station, TX, USA}

\def\L{\mathcal{L}}

\newcommand{\indep}{\perp \!\!\! \perp}

\begin{document}
	\maketitle
	\doublespacing
	
\begin{abstract}
This paper presents how to use common random number (CRN) simulation to evaluate Markov chain Monte Carlo (MCMC) convergence to stationarity. We provide an upper bound on the Wasserstein distance of a Markov chain to its stationary distribution after $N$ steps in terms of averages over CRN simulations. We apply our bound to Gibbs samplers on a model related to James-Stein estimators, a variance component model, and a Bayesian linear regression model. 
Using our examples, we show that the CRN-based simulation combined with a coalescing condition to generate a total variation bound converges to zero much more faster than the available drift and minorization bounds, while also converging at the same rate as the one-shot coupling bound.
\end{abstract}

\tableofcontents

	\section{Introduction}
 
Markov chain Monte Carlo (MCMC) algorithms
are often used to simulate from a stationary distribution
of interest (see e.g.~\cite{intromcmc}).
One of the primary questions when using these Markov chains is,
after how many
iterations is the distribution of the Markov chain sufficiently close to the
stationary distribution of interest,
i.e.\ when should actual sampling begin \cite{honestexp}.
The number of
iterations it takes for the distribution of the Markov chain to be
sufficiently close to stationarity is called the burn-in period.
Various informal methods are available for
estimating the burn-in period,
such as effective sample size estimation,
the Gelman-Rubin diagnostic, and visual checks
using traceplots or autocorrelation graphs
\cite{glm, mode, bayemethfin, convdiag}. While these methods ``can detect problems with an MCMC simulation, [...] they cannot prove that
the simulation is generating a representative sample." \cite{qin2024convergenceboundsmontecarlo}


Convergence analysis studies the
distance of a Markov chain to stationarity and is traditionally measured in terms of total
variation distance (e.g.~\cite{probsurv, tierney1994markov}),
though more recently
the Wasserstein distance has been considered
\cite{gibbswass, waslinearmixedmodels, wasandoneshot, wasmeth}.
However, finding upper bounds on
either distance can be quite difficult to establish \cite{intromcmc, honestexp, qin2024convergenceboundsmontecarlo}, and if an upper bound is known, it is usually based on
complicated problem-specific
calculations \cite{waslinearmixedmodels, geomconvrates, sixta, localcont}. 
This motivates the desire to instead estimate convergence bounds
from actual simulations of the Markov chain combined with less complicated calculations, which we consider here via common random number (CRN) simulation.




CRN simulation occurs when we initiate two copies of the chain with different values, but use the same sequence of random variables to simulate a specific random function representation of the chain (See Section \ref{sec:background} for a formal definition). This simulation method is also sometimes referred to as synchronous coupling \cite{sdesynch,browniansynch, jacob} and falls under the general framework of ``auxiliary
simulation" \cite{simconv}, i.e. using extra preliminary Markov chain runs to
estimate the convergence time needed in the final run.
Estimating Markov chain convergence rates using CRN simulation was first
proposed in \cite{crnconv} to find estimates of mixing times in total
variation distance. Further, one of the first Wasserstein bounds generated on a Markov chain model used CRN in its construction \cite{gibbswass}.
More recently, \cite{biswas} showed how CRN
simulation could be used for estimating an upper bound on the Wasserstein
distance (their Proposition~3.1), and provided useful
applications of the CRN method to high-dimensional and tall data (their
Section~4). CRN simulation to estimate convergence rates has been used on a wide range of applications including Langevin algorithms \cite{langevin}, stochastic gradient descent \cite{sgd}, stochastic differential equations (SDE) \cite{sdesynch}, and Hamiltonian Monte Carlo \cite{hmc, unbiasedhmc}.

The use of CRN to generate tight bounds on the Wasserstein and total variation distance is promising. In this paper we show that for the Gibbs sampler for a model related to James-Stein estimators and the variance component model, CRN bounds perform significantly better than drift and minorization (DnM) bounds (see figures \ref{fig:crnvsdnmstein} and \ref{fig:vcmcrnvsdnm} and commentary in Numerical Example \ref{numex:varcomp2} for more details). Under certain conditions, \cite{sdesynch} and \cite{crn} show that CRN is the optimal coupling of the Wasserstein distance between simulated solutions of SDEs and two random variables in the $\mathcal{L}^2$ metric, respectively. \cite{independentmetropolishastings} shows that CRN simulation on the independent Metropolis-Hastings algorithm is the maximal coupling of the total variation distance. \cite{browniansynch} shows that the distance between two reflected Brownian motion processes simulated using CRN converge to 0 almost surely.



In Theorem~\ref{thm:summary}, using CRN simulation, we provide an
estimated upper bound between the Wasserstein distance
of a Markov chain and the corresponding stationary distribution when
only the unnormalized density of the stationary distribution is known. Furthermore the $95\%$ confidence intervals of our estimates in the real data examples are fairly narrow. See Lemma \ref{lem:ci} and its implementation in the figures \ref{fig:simdiffstein}, \ref{fig:vcmcrn2}, \ref{fig:vcmcrnvsdnm} and \ref{fig:simdiff} and tables \ref{table:civarcompgalinjones} and \ref{table:civarcomprosecowles} for more details. 

This paper is organized as follows.
In Section~\ref{sec:background}, we present definitions
and notation.
In Section~\ref{sec:sim}, we establish convergence bounds of a
Markov chain to its corresponding stationary distribution using CRN
simulation when the initial distribution is not in stationarity.
In Section \ref{sec:stein}, we apply Theorem~\ref{thm:summary} to a Gibbs sampler for a model related to the James-Stein estimator and compare the CRN bound to a DnM and one-shot coupling bound in Real Data Example \ref{numex:stein}. 
In Section \ref{sec:varcomp} we apply Theorem~\ref{thm:summary} to a Gibbs sampler for a variance component model and compare the CRN bound to a DnM bound in Numerical Example \ref{numex:varcomp2} and another DnM bound in Real Data Example \ref{numex:varcomp3}. In Section \ref{sec-example} we apply Theorem~\ref{thm:summary} to a Gibbs sampler for a Bayesian linear regression model.

For all of the cases where we compare the CRN bound, which is measured in Wasserstein distance, to the DnM and one-shot coupling bound, which are measured in total variation distance, we must state a coalescing condition. The coalescing condition bounds the total variation distance with respect to the Wasserstein distance and is a part of a one-shot coupling bound \cite{sixta}. Further, since the one-shot coupling bound uses a contraction condition, which is established using CRN, we expect the empirical CRN simulation bound and the theoretical one-shot coupling bound to approach zero at the same rate (See Figure \ref{fig:crnvsdnmstein} as an example).
	
For each numerical example, we also provide the Gelman-Rubin diagnostic and traceplots for comparison (see Section \ref{subsec:convdiag}). 
The Gelman-Rubin statistic assesses convergence by comparing within-chain and between-chain variability across chains initialized from dispersed starting points. It is not immediately clear under what conditions the Gelman-Rubin statistic will produce a more conservative cutoff than the CRN bound. In three of the four numerical and real-data examples, the Gelman–Rubin statistic suggests an earlier cutoff than the CRN-based criterion. Also, since the initial distributions may vary (Gelman-Rubin statistic requires an overdispersed initial distribution), we are cautious in comparing the two methods. 
 Trace plots provide a visual diagnostic for assessing convergence by illustrating whether chains initialized at different starting values have converged to the same stationary distribution. The traceplots in our examples use the same initial distributions as our numerical examples. In all of the numerical and real-data examples, the traceplots indicate mixture sooner than the Gelman-Rubin statistic and the CRN bound.

The code used to generate all of the tables and calculations can be found at \\
\href{https://github.com/sixter/CommonRandomNumber/}{github.com/sixter/CommonRandomNumber/}.
	
	\section{Background and Notation}\label{sec:background}
	
	\paragraph{Metrics on distributions} Let $(\mathcal{X}, d)$ be a Polish metric space and $\mathcal{F}$ the associated Borel $\sigma$-algebra. The $L^p-$Wasserstein distance between distributions $\pi$ and $\nu$ on $(\mathcal{X},\mathcal{F})$ is defined as 
	$W_{d,p}(\pi,\nu)=\inf_{X\sim \pi, Y\sim \nu}E[d(X,Y)^p]^{1/p}$.
	The Wasserstein space $P_p(\mathcal{X})$ is the set of Borel probability measures on the metric space $(\mathcal{X},d)$ such that for any $\mu\in P_p(\mathcal{X})$, $\int_{\mathcal{X}}d(x_0,x)^p
	\mu (\mathrm{d}x) <\infty$ for some $x_0\in \mathcal{X}$.
Total variation distance is defined as $||\pi-\nu||_{TV}=\inf_{X\sim \pi, Y\sim \nu}P(X\neq Y)=\sup_{A\in \mathcal{F}}|\pi(A)-\nu(A)|$. The total variation distance can be written as a special case of the Wasserstein distance, $||\pi-\nu||_{TV}=W_{d_I,1}(\pi,\nu)$ where $d_I(X,Y)=I_{X\neq Y}$. If $\mathcal{X}\subset \mathbb{R}^q$, $q\in \mathbb{N}$, define the norm over a vector as $||x||_p = (\sum_{i=1}^q |x_i|^p)^{1/p}$ where $p\in [1,\infty)$. Define the p-norm over a function $\phi:\mathcal{X}\to \mathbb{R}$ as $\left\lVert\phi\right\rVert_{p}= \left(\int_{\mathcal{X}} \phi(x)^{p} \mathrm{d}x \right)^{\frac{1}{p}}$. If $X\indep Y$, then $X$ and $Y$ are independent. If $\pi\ll \nu$, then $\pi$ is absolutely continuous with respect to $\nu$.

\paragraph{CRN and Markov chains} Define a Markov chain $\{X_n\}_{n\geq 1}$ in $\mathcal{X}$ such that $X_n = k_{\theta_n}(X_{n-1})$, where  $\{\theta_n\}_{n\geq 1}$ are i.i.d. random variables on some measurable space $\Theta$ and random measurable mappings $k_{\theta_n} : \mathcal{X} \mapsto \mathcal{X}$. The set of random functions $k_{\theta_1}, k_{\theta_2},\ldots$ is called an iterated function system. Any time-homogeneous Markov chain can be represented as an iterated function system \cite{avcont}. We can also write $X_n=k_{\theta_n}(k_{\theta_{n-1}}(k_{\theta_{n-2}}\ldots k_{\theta_1}(x_0)))=k_{\Theta_n}(x_0)$ where  $\Theta_n=(\theta_1,\ldots, \theta_n)$ and $\theta_i\indep \theta_j$, $i\neq j$.  CRN simulation in this context occurs when we simulate two Markov chains $\{X_n\}_{n\geq 1}$ and $\{Y_n\}_{n\geq 1}$ with different initial values $x_0\neq y_0$, but with the same transition functions. That is, $\theta_n$ are the same in $X_n = k_{\theta_n}(X_{n-1})$ and $Y_n = k_{\theta_n}(Y_{n-1})$, $n\geq 1$. In different notation, the $\Theta_n$ is the same in $X_n=k_{\Theta_n}(x_0)$ and $Y_n=k_{\Theta_n}(y_0)$. See Algorithm \ref{alg:expdiff} for more details. The stationary distribution of $X_n$ is defined as $\mathcal{L}(X_{\infty})=\pi$.

\paragraph{Distribution functions}

$IG(\alpha, \beta)$ represents the inverse gamma distribution with  density function $f_{IG}(x \mid \alpha,\beta)= \frac{\beta^{\alpha}}{\Gamma(\alpha)}x^{-\alpha-1}e^{-\beta/x}$. $N(\mu, \sigma^2)$ represents the normal distribution with density function $f_{N}(x \mid \mu,
\sigma^2)= \frac{1}{\sqrt{2\pi\sigma^2}}e^{-\frac{(x-\mu)^2}{2\sigma^2}}$. For a distribution $\nu$, the density function is denoted as $f_{\nu}$.

	\section{Estimating Markov chain convergence rates using CRN simulation}\label{sec:sim}
 
We propose bounding the $L^p-$Wasserstein distance between the $n$th iteration of a Markov chain $\{X_n\}_{n\geq 1}$ and the corresponding stationary distribution $\pi$ through CRN simulation. Our main result is Theorem \ref{thm:summary}. Using CRN simulation as a convergence diagnostic tool was first discussed in \cite{crnconv} for bounding total variation distance.

We first provide a method for bounding the Wasserstein distance between $\mathcal{L}(X_n)$ and the corresponding stationary distribution $\pi$,  $W_{d,p}(\mathcal{L}(X_n),\pi)$ by proposing an auxiliary random variable $Y_0\sim \nu$. That is $W_{d,p}(\mathcal{L}(X_n),\pi)$ is bounded above by the product of the `distance' between $\pi$ and $\mathcal{L}(Y_0)$ ($K_{s}(\pi,\nu)$, $s>1$) and $E[d(X_n,Y_n)]$. 

%

\begin{theorem}\label{thm:1}
	Let $\{X_n\}_{n\geq 1}$ and $\{Y_n\}_{n\geq 1}$ be two copies of a Markov chain with stationary distribution $\pi$ and initial distributions $\mathcal{L}(Y_0)= \nu$, $\mathcal{L}(X_0)= \mu$ where $X_0\indep Y_0$. Assume $\pi\ll\nu$.
	Then for $r\geq p\geq 1$ and $s> 1$, the following holds where   $K_{s}(\pi,\nu)=E_{Y_0\sim\nu}\left[\left( \frac{f_{\pi}(Y_0)}{f_{\nu}(Y_0)}\right)^{\frac{s}{s-1}}\right]^{\frac{s-1}{s}}$
	\begin{align}\label{eq:thm1}
		W_{d,p}(\mathcal{L}(X_n),\pi)\leq	K_{s}(\pi,\nu)^{1/r} E\left[d(X_n,Y_n)^{rs}\right]^{\frac{1}{r s}}
	\end{align}
	See Section \ref{pf:thm1original} for a proof.
\end{theorem}


\begin{remark}\label{rem:inftynorm}
	If we take $K_{1}(\pi,\nu)=\lim_{s\downarrow 1}E_{Y_0\sim\nu}\left[\left( \frac{f_{\pi}(Y_0)}{f_{\nu}(Y_0)}\right)^{\frac{s}{s-1}}\right]^{\frac{s-1}{s}}$, then from Theorem \ref{thm:1} we get that
	$K_{1}(\pi,\nu)= ess \sup_{x\in\mathcal{X}} f_{\pi}(x)/f_{\nu}(x)$ and
	\begin{align*}\label{eq:reminftynorm}
		W_{d,p}(\mathcal{L}(X_n),\pi)\leq	K_{1}(\pi,\nu)^{1/r} E\left[d(X_n,Y_n)^{r} \right]^{1/r}
	\end{align*}
	The case of $s\downarrow 1$ can be proven using the rejection sampler or the separation distance (see the proof of this special case in Section \ref{pf:thm1}).  The use of rejection sampling to generate an upper bound on the expected distance between a proposal density function $\nu$ and a stationary density function $\pi$ was motivated by \cite{crnconv}, which used rejection sampling to generate similar upper bounds in total variation distance. This case is used in Lemmas \ref{lem:steinKbound} and \ref{lem:regKbound}
\end{remark}

\begin{remark}\label{rem:tv}
	Applying $d_I(X,Y)=I_{X\neq Y}$ to Remark \ref{rem:inftynorm}, and setting $p=r=1$, the total variation distance between a Markov chain and the corresponding stationary distribution is bounded above as follows,
	\begin{equation}\label{eqn:tvbound}
		||\mathcal{L}(X_n)-\pi||_{TV}\leq K_{1}(\pi,\nu)P(X_n\neq Y_n)
	\end{equation}
	And so, for stopping time $\tau = \min\{n:d(X_n,Y_n)=0\}$, $||\mathcal{L}(X_n)-\pi||_{TV}\leq K_{1}(\pi,\nu)P(\tau > n)$. 
\end{remark}

The following are additional observations regarding $K_{s}(\pi,\nu)$.
\begin{enumerate}
	\item $K_{s}(\pi,\nu)$ is related to the $f$-divergence between the initialized and target distribution. The $f$-divergence suffers from the curse of dimensionality \cite{fdiv,fdiv2} and so $K_{s}(\pi,\nu)$ may become very large as dimension increases. 
	\item 	In many cases, the normalizing constant for $f_{\pi}$ is unknown. That is, we only know of a function $g(x)$ such that $\frac{1}{L} g(x)=f_{\pi}(x)$ where $L$ is unknown. For cases like this, we note that for all $x\in\mathcal{X}$ and any $B\subset \mathcal{X}$,
	
	\begin{equation}\label{eqn:kupper}
		f_{\pi}(x) = \frac{g(x)}{\int_{\mathcal{X}}g(x)\mathrm{d}x}\leq \frac{g(x)}{\int_{B}g(x)\mathrm{d}x}
	\end{equation}
	
	This upper bound on $f_{\pi}$ is used to upper bound $K$ in all of the examples, which we calculate analytically.
	\item 	The quantity $K_{s}(\pi,\nu)$ can also be estimated using importance sampling, since  $\left(\frac{f_{\pi}(Y_0)}{f_{\nu}(Y_0)}\right)^{\frac{s}{s-1}}$ typically becomes large on low-probability events under $\nu$. As $K_{2}(\pi,\nu)$ increases, the more computationally expensive it will be to calculate \cite{impsampling}. 
\end{enumerate}
%

Next we define Algorithm \ref{alg:expdiff}, which generates an estimate of $E[d(X_n,Y_n)^{rs}]$ using CRN simulation.

\begin{algorithm}
	\caption{An estimate using CRN of $E[d(X_N,Y_N)^{rs}]\approx \frac{1}{M}\sum_{i=1}^M d(x_{N,i},y_{N,i})^{rs}$}\label{alg:expdiff}
	\begin{algorithmic}
		\For{$i = 1, \ldots, M$}
		\State $x_{0,i} \sim \mu$, $y_{0,i} \sim \nu$ where $x_{0,i} \indep y_{0,i}$
		\For{$n =1, \ldots,N$}
		\State $\theta_{n} \sim \mathcal{L}(\Theta)$
		\State $x_{n,i} \gets k_{\theta_{n}}(x_{n-1,i})$
		\State $y_{n,i} \gets k_{\theta_{n}}(y_{n-1,i})$
		\EndFor
		\EndFor
		\State \textbf{return} $\frac{1}{M}\sum_{i=1}^M d(x_{N,i}-y_{N,i})^{rs}$
	\end{algorithmic}
\end{algorithm}

Combining Algorithm \ref{alg:expdiff} with Theorem \ref{thm:1}, we can simulate an upper bound on the Wasserstein distance between the law of a Markov chain at iteration $N$, $\mathcal{L}(X_N)$, and the corresponding stationary distribution, $\pi$, as follows.

    \begin{theorem}\label{thm:summary}
			Let $\{X_n\}_{n\geq 1}$ and $\{Y_n\}_{n\geq 1}$ be two copies of a Markov chain with stationary distribution $\pi$, initial distributions $\mathcal{L}(Y_0)= \nu$, $\mathcal{L}(X_0)= \mu$ ($X_0\indep Y_0$) and simulated according to Algorithm \ref{alg:expdiff}.
		Assume $\pi\ll\nu$. Then for $r\geq p\geq 1$ and $s\geq 1$, the following holds almost surely for $N\geq 1$, where  $K_{s}(\pi,\nu)=E\left[\left( \frac{f_{\pi}(Y_0)}{f_{\nu}(Y_0)}\right)^{\frac{s}{s-1}}\right]^{\frac{s-1}{s}}$ and  $\mathcal{L}(X_N),\mathcal{L}(Y_N)\in P_{rs}(\mathcal{X})$.
		\begin{equation}\label{eqn:summary}
				W_{d,p}(\mathcal{L}(X_N),\pi) 
				\leq K_{s}(\pi,\nu)^{1/r}\left(\lim_{M\to \infty}\frac{1}{M}\sum_{i=1}^M d(X_{N,i},Y_{N,i})^{r s}\right)^{\frac{1}{rs}}
			\end{equation}
		\begin{proof}
		Since $\mathcal{L}(X_N),\mathcal{L}(Y_N)\in P_{rs}(\mathcal{X})$, $E[d(X_N,Y_N)^{rs}] < \infty$ and so by the strong law of large numbers, $\lim_{M\to \infty} \frac{1}{M}\sum_{i=1}^M d(X_{N,i},Y_{N,i})^{rs}\overset{a.s.}{=} E[d(X_N,Y_N)^{rs}]$. 
			By Theorem \ref{thm:1}, $W_{d,p}(\mathcal{L}(X_N),\pi) 
			\leq K_{s}(\pi,\nu)^{1/r} E[d(X_N,Y_N)^{rs}]^{\frac{1}{rs}}\overset{a.s.}{=} K_{s}(\pi,\nu)^{1/r} \left(\lim_{M\to \infty}\frac{1}{M}\sum_{i=1}^M d(X_{N,i},Y_{N,i})^{rs}\right)^{\frac{1}{rs}}$
		\end{proof}
	\end{theorem}

To optimize this bound we can choose $\nu$ that is close to $\pi$. If $K_{s}(\pi,\nu)$ is very large (which may happen in high dimension), we can choose an $r\geq p$ such that $K_{s}(\pi,\nu)^{1/r}$ is a reasonable value. However, choosing a larger $r$ increases the value of $E[d(X_n,Y_n)^{rs}]^{\frac{1}{rs}}$, resulting in a trade-off between the two effects. Nevertheless, when the CRN simulation produces expected values close to those of an optimal coupling, the distance $d(X_N, Y_N)$ is often small, and the resulting CRN-based bound may still perform well empirically even in high-dimensional settings with large r.

Next we state the confidence interval on our bound. We use this confidence interval in figures \ref{fig:simdiffstein}, \ref{fig:vcmcrn2}, \ref{fig:vcmcrnvsdnm} and \ref{fig:simdiff} and Table \ref{table:civarcompgalinjones}.
\begin{lemma}\label{lem:ci}
	The $1-\alpha$ confidence interval on the Wasserstein upper bound, based on a simulation estimate of $E[d(X_N,Y_N)^{rs}]^{\frac{1}{rs}}$ is given by,
	\begin{equation*}
		\left[K_{s}(\pi,\nu)^{\frac{1}{r}}\left(m - z_{\alpha/2}\sigma_m/\sqrt{M}\right)^{\frac{1}{rs}}, K_{s}(\pi,\nu)^{\frac{1}{r}}\left(m + z_{\alpha/2}\sigma_m/\sqrt{M}\right)^{\frac{1}{rs}}\right]
	\end{equation*}
	Where $m = \frac{1}{M}\sum_{i=1}^M d(X_{N,i},Y_{N,i})^{rs}$, $\sigma^2_m$ is the sample variance of $m$ and  $\mathcal{L}(X_{N}),\mathcal{L}(Y_{N})\in P_{2rs}(\mathcal{X})$.
	\begin{proof}
	Since $\mathcal{L}(X_{N}),\mathcal{L}(Y_{N})\in P_{2rs}(\mathcal{X})$, the variances are finite and we can apply the Central Limit Theorem. That is,
	$P\left(m - \frac{z_{\alpha/2}\sigma_m}{\sqrt{M}}\leq E[d(X_n,Y_n)^{rs}] \leq  m + \frac{z_{\alpha/2}\sigma_m}{\sqrt{M}}\right)\approx 1-\alpha$. 
	And so, 
	$P\left(K_{s}(\pi,\nu)^{\frac{1}{r}}\left(m -\frac{z_{\alpha/2}\sigma_m}{\sqrt{M}} \right)^{\frac{1}{rs}}
	\leq K_{s}(\pi,\nu)^{\frac{1}{r}}E[d(X_n,Y_n)^{rs}]^{\frac{1}{rs}} \leq  
	K_{s}(\pi,\nu)^{\frac{1}{r}} \left(m + \frac{z_{\alpha/2}\sigma_m}{\sqrt{M}}\right)^{\frac{1}{rs}}\right)\approx1-\alpha$.
	\end{proof}
\end{lemma}
Note that since we are applying the Central Limit Theorem, this is an asymptotic confidence interval with respect to $M\to \infty$.
\paragraph{The approach of Biswas and Mackey \cite{biswas}} Algorithm \ref{alg:expdiff} closely resembles Algorithm 1 of \cite{biswas}. Biswas and Mackey show in \cite{biswas} how CRN simulation can be used to estimate the CUB estimator, which is defined as follows for a given $p\in [1,\infty)$,

\begin{equation}
	\text{CUB}_{M,N,T}=\left(\frac{1}{M(T-N)}\sum_{i=1}^M \sum_{n=N+1}^T d(X_{n,i},Y_{n,i})^p \right)^{1/p}
\end{equation}

By Corollary 3.2 of \cite{biswas}, under certain regularity conditions, the CUB estimator bounds the Wasserstein distance between the sample means,
$$W_{d,p}\left(\mathcal{L}\left(\frac{1}{(T-N)} \sum_{n=N+1}^T X_{n,i} \right), \mathcal{L}\left(\frac{1}{(T-N)} \sum_{n=N+1}^T Y_{n,i} \right)\right) \leq \lim_{M\to \infty} \text{CUB}_{M,N,T}$$

If we set $N,T-N\to \infty$ we obtain a bound on the stationary distributions,
$$W_{d,p}(\mathcal{L}(X_{\infty}),\mathcal{L}(Y_{\infty})) \leq \lim_{M,N\to\infty,T-N\to \infty} \text{CUB}_{M,N,T}$$ 



In \cite{biswas}, the CUB estimator eventually bounds the Wasserstein distance between the corresponding stationary distributions, but does not tell us precisely when it will happen or by how much. That is, if the starting points of the two Markov chains are close in distance but are far from stationarity, the results from \cite{biswas} would conclude, for some fixed iteration $N$, that the distance between the two chains is close (true), but it would be incorrect to state that the distance to stationarity is close. In comparison, Theorem \ref{thm:summary} measures distance to stationarity for any initial distribution on $X_0$ and iteration $N\in \mathbb{N}$.  

\paragraph{The approach of Johnson \cite{crnconv}} The motivation of Theorem \ref{thm:summary} comes from \cite{crnconv}, which provides a bound in total variation distance between a Markov chain and its corresponding stationary distribution using CRN simulation. The following is a simplification of the results with altered notation to match this paper. For complete details, refer to \cite{crnconv}. For a discussion on when CRN is applicable within the setting of \cite{crnconv}, see \cite{mcconvcomp}.

Let $\tau$ be a `stopping' time, $\tau = \min\{n:d(X_n,Y_n)\leq \epsilon\}$. 
Theorem 2 of \cite{crnconv} states that the total variation between two Markov chains $\{X_n\}_{n\geq 1}$, $\{Y_n\}_{n\geq 1}$ both with initial distribution $\nu$ and stationary distribution $\pi$ is as follows,

$$||\pi-\mathcal{L}(X_n)||_{TV}\leq 2 K_{1}(\pi,\nu) P(\tau\geq n) + O(\epsilon)$$


The notation $O(\epsilon)$ means that for some fixed  $\epsilon_0>0$, there exists an $M\in\mathbb{N}$ such that for any $\epsilon\in (0,\epsilon_0)$, $||\pi-\mathcal{L}(X_n)||_{TV}\leq \frac{2 P(\tau\geq n)}{1-r(\pi,\nu)} + M\epsilon$. 

In \cite{crnconv}, $K_{1}(\pi,\nu)=\frac{1}{1-r(\pi,\nu)}$ where $r(\pi,\nu)$ is the rejection rate in a rejection sampler with stationary distribution $\pi$ and proposal distribution $\nu$ (See Definition  \ref{def:K}). By Equation \ref{def:K}, $\frac{1}{1-r(\pi,\nu)}=ess\sup_z \frac{f_{\pi}(z)}{f_{\nu}(z)}$. Thus, both Theorem \ref{thm:summary} (with $s_1=1,s_2=\infty$) and Theorem 2 of \cite{crnconv} require calculating upper bounds on $r(\pi,\nu)$. However, Theorem \ref{thm:summary} estimates an exact upper bound on the Wasserstein distance while Theorem 2 of \cite{crnconv} estimates an upper bound on the total variation distance with an unknown component that converges to $0$ as $\epsilon\to 0$.

Applying total variation distance to the bound generated by Theorem \ref{thm:summary} and taking $\epsilon=0$, we get a similar result to \cite{crnconv}, $||\pi-\mathcal{L}(X_n)||_{TV}\leq K_{1}(\pi,\nu) P(\tau\geq n) $. See Remark \ref{rem:tv} for more details. In our simulations, exact simulation of $d(X_n,Y_n)=0$ was not realistically attainable, however. This is why \cite{crnconv} introduces the error $\epsilon$. This points to CRN simulation being a better fit for bounding Wasserstein distance compared to total variation distance.

\section{Gibbs sampler for a model related to James-Stein estimators}\label{sec:stein}

Consider a Gibbs sampler for a model related to the James-Stein estimator, which we will apply Theorem \ref{thm:summary} to simulate convergence bounds on the Wasserstein distance. The James-Stein estimator is similar to a variance component model (Example \ref{ex:varcomp}), but each $\theta_i$ has only one observation, $Y_i$. Its significance in the frequentist statistics literature comes from Stein's paradox, which proves that for $\theta_i$ the estimator, $\left(1-\frac{I-2}{\sum_{j=1}^I Y_j^2}\right)_+Y_i$, has smaller risk than the intuitive estimator, $Y_i$, when $I\geq 3$ \cite{wasserman}. Here, we look at the same model setup, but through a Bayesian statistics lense. See \cite{stein} for more details on this example. 

\begin{example}[A model related to James-Stein estimators]\label{ex:stein}
	
	Suppose that we have the following observed data $Y\in \mathbb{R}^q$ where
	$$Y_i |\theta_i\sim N(\theta_i,V)$$
	for known $V>0$ and unknown $\theta_i\sim N(\mu,A)$ ($\theta_i$'s are i.i.d.), where $\mu$ has a flat prior and $A\sim IG(\alpha,\beta)$.
	
	The joint posterior density function of $\vec{\theta},\mu,A\mid Y$ is proportional to the following equation,
	\begin{align}\label{eqn:jsteinG}
		g(\vec{\theta},\mu,A\mid Y) = f_{IG}(A\mid \alpha, \beta) \times \prod_{i=1}^q f_N(Y_i\mid \theta_i, V)  \times \prod_{i=1}^q f_N(\theta_i \mid \mu, A)
	\end{align}		
\end{example}
\paragraph{The Gibbs sampler for a model related to James-Stein estimators.} 

Algorithm \ref{alg:stein} shows how to apply a Gibbs sampler on this example.
\begin{algorithm}\label{alg:stein}
	\caption{The Gibbs sampler algorithm for a model related to James-Stein estimators.}\label{alg:stein}
	\begin{algorithmic}
		\State Initialize $(\vec{\theta}_0,\mu_0,A_0)$
		\For{$n = 1, \ldots, N$}
		\State For each $i\in 1,\ldots, q$, generate $\theta_{n,i}|\mu_{n-1}, A_{n-1},Y \sim N(\frac{Y_{i}A_{n-1}+\mu_{n-1} V}{A_{n-1}+V},\frac{VA_{n-1}}{V+A_{n-1}})$
		\State Generate $\mu_n| \vec{\theta}_{n},A_{n-1},Y \sim N(\bar{\theta}_{n},A_{n-1}/q)$ where $\bar{\theta}_n=\frac{1}{q}\sum_{i=1}^q \theta_{n,i}$
		\State Generate $A_n|\vec{\theta}_{n},\mu_{n},Y \sim IG\left(\alpha+ \frac{q-1}{2},\frac{\sum_{i=1}^q(\theta_{n,i}-\mu_{n})^2}{2}+\beta\right)$
		\EndFor
	\end{algorithmic}
\end{algorithm}

We will bound the Wasserstein distance between a Markov chain and its corresponding stationary distribution using Lemma \ref{lem:steinKbound}. 


\begin{lemma}\label{lem:steinKbound}    
	Let $\{X_n\}_{n\geq 1}=(\vec{\theta}_n,\mu_n,A_n)_{n\geq 1}$ and $\{X'_n\}_{n\geq 1}=(\vec{\theta}'_n,\mu'_n,A'_n)_{n\geq 1}$ be two copies of the Markov chain initialized with $\vec{\theta}'_0 \sim N_q(\vec{Y},V I_q)$, $A'_0\sim IG(\alpha+(q-1)/2,\beta)$, $\mu'_0\sim N(\bar{\theta}'_0,A'_0)$, and $X_0\indep X'_0$. Fix $p\in(1,\infty]$ and l
	et $\pi$ be the corresponding stationary distribution. 
	Then for $r\geq p$
	\begin{align*}
		W_{d,p}(\mathcal{L}(X_n),\pi)
		\leq \left(\frac{\Gamma(\alpha+(q-1)/2)(2\pi)^{1/2}}{\beta^{\alpha+(q-1)/2}}
		\frac{1}{L}
		E[d(X_n,X'_n)^r]\right)^{1/r}
	\end{align*}
	
	Where $L=	\int_{B} \frac{A^{-\alpha-1}e^{-\beta/A}}{(A+V)^{q/2}}
	e^{\frac{1}{2}\sum_{i=1}^q\left[\frac{(\frac{Y_i}{V}+\frac{\mu}{A})^2}{\frac{1}{V}+\frac{1}{A}}-\frac{Y_i^2}{V}-\frac{\mu^2}{A}\right]}\mathrm{d}(\mu, A)$ and $B\subset\mathbb{R}\times\mathbb{R}_+$ is a bounded set. 
	
	The proof is in section \ref{pf:steinKbound}.
\end{lemma}

\begin{numexmp}\label{numex:stein}
	We apply our results to the batting averages ($Y$) for 18 ($q=18$) major league baseball players in 1970 \cite{baseballdata, pscl}. We want to find the estimated upper bound on the Wasserstein and total variation distance for a Gibbs sampler applied to the James-Stein estimator using the baseball data. We set the priors to $\alpha=0.01, \beta=2$ and $V=Var(Y)=0.00485$. The initial value of our Markov chains are independently drawn;
	$X_0=(\vec{\theta}_0,\mu_0,A_0)=(100)^{q+2}$ and  $X'_0$ is distributed according to Lemma \ref{lem:steinKbound}. Using Lemma \ref{lem:steinKbound}, $K_{1}(\pi,\nu)=5.9535$ and 
	$$W_{||\cdot||_1,1}(\mathcal{L}(X_n), \pi)\leq  5.9535 E[||X_n-X'_n||_1]$$ holds almost surely.
	
	Using the CRN technique, we simulated $||X_n-X'_n||_1$ one thousand times ($M=1000$) over 20 iterations ($N=20$). At iteration 1, $\frac{1}{1000}\sum_{i=1}^{1000}||X_{1,i}-X'_{1,i}||_1=1994.712$. By iteration 9, $\frac{1}{1000}\sum_{i=1}^{1000}||X_{9,i}-X'_{9,i}||_1=0.00012$. Figure \ref{fig:simdiffstein} graphs the upper bound of the Wasserstein distance between $\mathcal{L}(X_n)$ and $\pi$ on a log scale. 
	\begin{figure}
		\centering
		\title{Upper bound on $W_{||\cdot||_1,1}(\mathcal{L}(X_n),\pi)$}
		\includegraphics[width=0.7\linewidth]{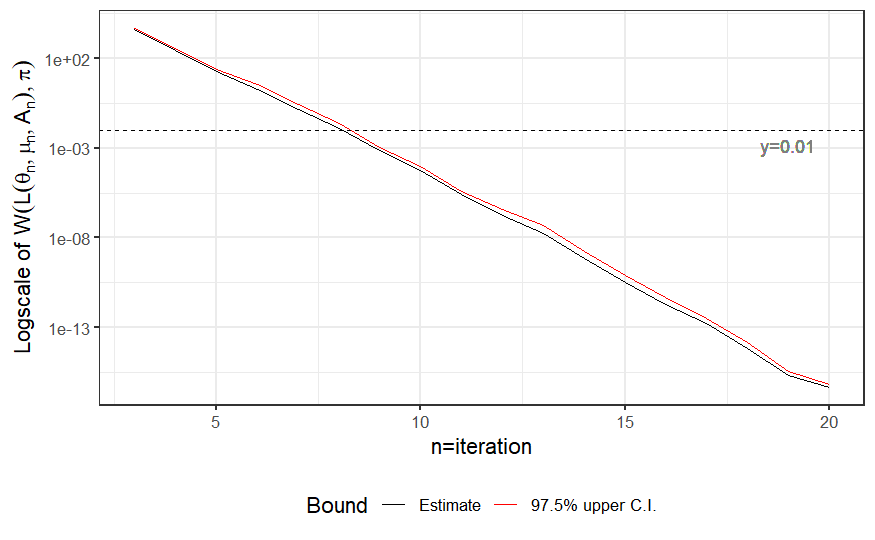}
		\caption{The upper bound on the Wasserstein distance for the Gibbs sampler model related to James-Stein estimators. The estimate uses the mean of 1000 simulations and initial value $(\vec{\theta}_{0},\mu_{0},A_{0})=(100)^{q+2}$. 
			The vertical axis is log scaled. $r=1$.}
		\label{fig:simdiffstein}
	\end{figure}
	At iteration 9, the bound on the Wasserstein distance first hits below 0.01 and $W_{||\cdot||_1,1}(\mathcal{L}(X_9), \pi)\leq 0.00073$.

	Before we bound the Markov chain and the corresponding stationary distribution in total variation we introduce Theorem 4.8 of \cite{sixta}, which provides an upper bound in total variation based on the expected differences between the Markov chains through the constant $K_{TV}$. This theorem will also be used to generate the one-shot coupling bound. 
	
	\begin{prop}[Theorem 4.8 of \cite{sixta}]\label{thm:anothergibbssampler}
		For two copies of a Gibbs sampler on a model related to the James-Stein estimators with flat priors on $\mu$ and $A$, the following holds for $q\geq 3$,
		\begin{align}\label{eqn:anothergibbssamplereqn}
			\|\L(X_{n})-\L(X'_n)\|_{TV}&\leq K_{TV} E[|A_0-A'_0|] \left(\frac{1}{q}\right)^{n-1}
		\end{align}
		where $K_{TV}= \frac{(S/2)^{\frac{q-1}{2}}}{\Gamma(\frac{q-1}{2})}\left(\frac{S}{q+1}\right)^{-\frac{q-3}{2}}e^{-\frac{q+1}{2}}$ and $S=\sum_{i=1}^q (Y_i-\bar{Y})^2$.
	\end{prop}
	
	Using Proposition \ref{thm:anothergibbssampler}, we get $K_{TV}=0.0047$. Therefore by Lemma 4.11 of \cite{sixta},
	\begin{align*}
		\|\L(X_{n})-\pi\|_{TV}
		\leq  0.0282 E[||X_n-X'_n||_1]
	\end{align*} 
	So, by iteration 6, $\|\L(X_{6})-\pi\|_{TV}\leq 0.00866$.
	
	We compare our bound to Proposition \ref{thm:anothergibbssampler} and Theorem 4 of \cite{stein}, which use one-shot coupling and DnM arguments, respectively, to generate a theoretical bound in total variation. The methodologies and assumptions are outlined in Table \ref{table:steinfunc}.
	
	\begin{table}[H]
		\centering
		\begin{tabular}{|p{1.4cm}|p{4.5cm}|p{4.6cm}|p{4.3cm}|}
			\hline
			\textbf{Method} & Common random number & Drift and minorization & One shot coupling \\
			\hline
			\textbf{Source} & Lemma \ref{lem:steinKbound} & Theorem 4 of \cite{stein} & Proposition \ref{thm:anothergibbssampler} and Proposition 6 of \cite{wasmeth}\\
			\hline
			\textbf{Prior} & $A\sim IG(0.01,2)$ and $f_{\mu}\propto 1$ & $A\sim IG(-1,2)$ and $f_{\mu}\propto 1$ & $f_{A}\propto 1$ and $f_{\mu}\propto 1$\\
			\hline
			\textbf{Bound} & \makecell[l]{$\|\L(X_n)-\pi\|_{TV}\leq$\\$ 0.0282 E[|X_n-X'_n|]$ \\ where $X_n,X'_n$ are \\ simulated using CRN.}& \makecell[l]{$\|\L(X_n)-\pi\|_{TV}\leq$\\$0.967^n+ (0.935)^n\times$\\$ (1.17+\sum_{i=1}^q(\theta_i-\bar{Y})^2)$} & \makecell[l]{$\|\L(X_n)-\pi\|_{TV}\leq$\\$\frac{0.0282}{17} E[|A_0-A_1|] \left(\frac{1}{18}\right)^{n} $} \\
			\hline
		\end{tabular}
		\caption{Theoretical bounds that are benchmarked against the CRN simulated bound. All bounds have initial value $X_0 = (\vec{\theta}_{0}, \mu_{0},A_{0})=(100)^{q+2}$. $f$ denotes the density function.}
		\label{table:steinfunc}
	\end{table}
	Note the minor variations in the examples, Proposition \ref{thm:anothergibbssampler} assumes that the prior on $A$ is flat. In Theorem 4 of \cite{stein}, the prior on $A\sim IG(-1,2)$ with the intent of generating a flat distribution. In our real data example we set the prior to $A\sim IG(0.01, 2)$, which could approximate a flat distribution. We will assume the minor variations are comparable. 
	
	Table \ref{table:steinvals} and Figure \ref{fig:crnvsdnmstein} compare the bounds on total variation distance using three different methods: CRN, DnM, and one-shot coupling. The one-shot coupling bound indicates total variation distance of less than 0.01 by iteration 5. The DnM bound indicates total variation distance of less than 0.01 by iteration 249. Thus, the one-shot coupling bounds performs the best closely followed by the CRN bound and significantly better than the DnM bound.
	\begin{table}
		\centering
		\begin{tabular}{|c|c|c|c|}
			\hline
			Iteration & \textbf{CRN} & \textbf{DnM} & \textbf{One-Shot Coupling}  \\
			\hline
			1 & 56.2430 & 167,409.95 & 58.6595 \\
			2 & 336.0368 & 156,528.33& 3.2589 \\
			3& 20.0791 & 146,354.02& 0.1811 \\
			4& 1.4003 & 136,841.04& 0.0101 \\
			5& 0.0898 & 127,946.40& \textbf{0.0005} \\
			6& \textbf{0.0087} & 119,629.91& 3.1049e-05		 \\
			7& 0.0007 & 111,853.99& 1.7247e-06	 \\
			8& 6.0301e-05 & 104,583.51& 9.5814e-08	 \\
			9& 3.4421e-06 & 97,785.60& 5.3230e-09	 \\
			10 & 2.6439e-07 & 91,429.56& 2.9572e-10 \\
			\ldots & \ldots & \ldots & \ldots\\
			249 & 1.6957e-307 & \textbf{0.0099} & 2.7286e-310\\
			\hline
		\end{tabular}
		\caption{Comparison of total variation bounds to stationarity for methods outlined in Table \ref{table:steinfunc} of the Gibbs sampler for a model related to James-Stein estimators. The values for each method that first reach less than 0.01 are in bold.}
		\label{table:steinvals}
	\end{table}
	
	\begin{figure}
		\centering
		\title{Comparison of bounds on $||\mathcal{L}(\vec{\theta}_n,\mu_n, A_n)-\pi||_{TV}$}
		\includegraphics[width=0.7\linewidth]{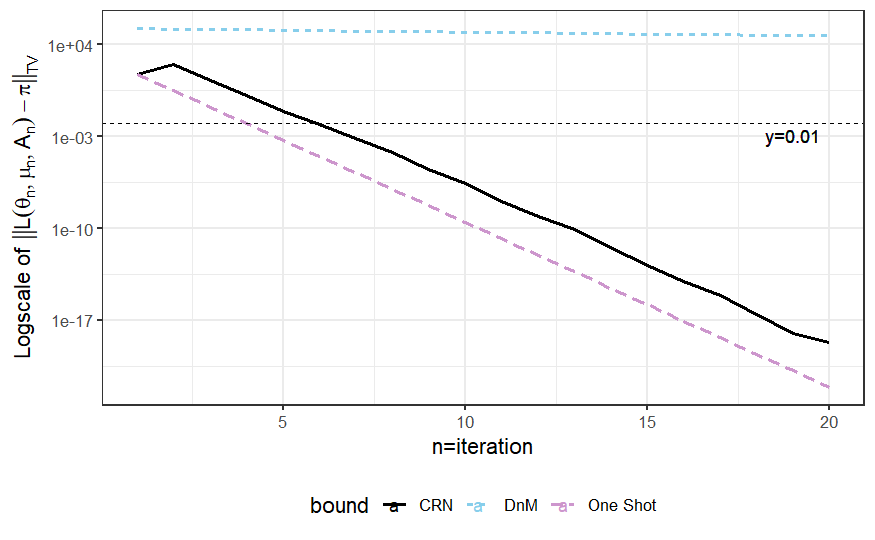}
		\caption{Comparison of an upper bound on the total variation distance on the Gibbs sampler for a model related to James-Stein estimators (Real Data Example \ref{numex:stein}) using a CRN simulated bound, a DnM bound and a one-shot coupling bound. The CRN estimate uses the mean of 1000 simulations. For all of the bounds, $(\vec{\theta}_{0},\mu_{0},A_{0})=(100)^{q+2}$. }
		\label{fig:crnvsdnmstein}
	\end{figure}
	
	
	The bound generated from CRN simulation and one-shot coupling are expected to be very similar since to generate the geometric convergence rate of the one-shot coupling bound (the $1/18$ in table \ref{table:steinfunc}) we use CRN (see \cite{jacob, sixta, localcont, avcont}). This is confirmed by Figure \ref{fig:crnvsdnmstein} where the CRN bound and the one-shot bound are approximately parallel. 
	
	%
	
	Finally, we compare our convergence rates with informal methods in Section \ref{subsec:convdiagstein}. The Gelman-Rubin statistic is evaluated for each parameter. By iteration 5, the Gelman-Rubin statistic for all of the parameters graphed is less than 1.05.  The traceplot paints a similar picture suggesting that all parameters have converged by iteration 6. Thus, the CRN bound indicating total variation distance convergence by iteration $6$, is tight. 
\end{numexmp}

\section{Gibbs sampler for a variance component model}\label{sec:varcomp}

Consider a Gibbs sampler on the variance component model (see \cite{varcomp}), which we will apply Theorem \ref{thm:summary} to simulate convergence bounds on the Wasserstein distance. Convergence rate bounds on the total variation metric for the variance component model have already been discussed in Theorems 1 of \cite{varcomp,Jones2004SufficientBF,simconv}. 

\begin{example}[Variance component model a.k.a. random effects model a.k.a. hierarchical model]\label{ex:varcomp}
	Suppose that we have a population that has overall mean $\mu$, which consists of $I$ groups of mean $\theta_i \sim N(\mu, V)$ and where each group has $J_i$ observations with distribution $Y_{ij}\sim N(\theta_i,W)$. Conditional on $\vec{\theta}$ and $W$, the $Y_{ij}$'s are independent. Conditional on $\mu$ and $V$, the $\theta_{i}$'s are independent. The parameters $\mu, \vec{\theta}, V,W$ are unknown. See Figure \ref{schema:varcomp} for a graphical representation.
	\begin{figure}
		\centering
		\begin{tikzpicture}
			[
			level 1/.style={sibling distance=25mm},
			level 2/.style={sibling distance=15mm},
			]
			\node {$\mu$}
			child {
				node {$\theta_1$}
				child {node {$Y_{1,1},\ldots Y_{1,J_{1}}$}}
			}
			child {
				node {$\theta_2$}
				child {node {$Y_{2,1},\ldots Y_{2,J_{2}}$}}
			}
			child {node {$\ldots$}  edge from parent[draw=none]}
			child {
				node {$\theta_I$}
				child {node {$Y_{I,1},\ldots Y_{I,J_{I}}$}}
			};
		\end{tikzpicture}
		\caption{Schema of the variance component model where $\theta_i \sim N(\mu, V)$, $Y_{ij}\sim N(\theta_i,W)$. The parameters $W,V,\mu$ are unknown. $\theta_i|\mu, V \indep \theta_j|\mu, V$ and $Y_{i_1,j_1}|\vec{\theta},W \indep Y_{i_2,j_2}|\vec{\theta},W$. }
		\label{schema:varcomp}
	\end{figure}
\end{example}
Propose the following priors on the unknown parameters where $a_1,b_1,a_2,b_2, b_3>0$.

\begin{equation}\label{eq:par}
	V\sim IG(a_1, b_1) \hspace{1cm} W\sim IG(a_2,b_2) \hspace{1cm} \mu \sim N(a_3, b_3)
\end{equation}

The unnormalized posterior distribution of $(\vec{\theta}, V,W, \mu)$ is written as follows,
\begin{align*}
	&g(\vec{\theta}, v, w, \mu) \\&= f_{IG}(V\mid a_1, b_1) f_{IG}(W\mid a_2, b_2) f_N(\mu \mid a_3, b_3)
	\left( \prod_{i=1}^{I}f_N(\theta_i \mid \mu, V)\right)
	\left( \prod_{i=1}^{I} \prod_{j=1}^{J_i} f_{N}(Y_{ij}\mid \theta_i, W)\right)
\end{align*}

We would like to sample from the normalized posterior distribution. 
\paragraph{The Gibbs sampler for the variance component model.}

Algorithm \ref{alg:varcomp} shows how to apply a Gibbs sampler on this example.

\begin{algorithm}
	\caption{A Gibbs sampler algorithm for the variance component model}\label{alg:varcomp}
	\begin{algorithmic}
		\State Initialize $(\vec{\theta}_0, V_0, W_0, \mu_0)=(\vec{\theta}_0, v_0, w_0, \mu_0)$
		\For{$n = 1, \ldots, N$}
		\State Generate $W_n |\vec{\theta}_{n-1} \sim IG\left(a_2+JI/2,b_2 + \frac{\sum_{i=1}^I\sum_{j=1}^J(Y_{ij}-\theta_{n-1,i})^2}{2}\right)$
		\State Generate $V_n| \mu_{n-1}, \vec{\theta}_{n-1} \sim IG\left(a_1+I/2,b_1 + \frac{\sum_{i=1}^I(\theta_{n-1,i}-\mu_{n-1})^2}{2}\right)$
		\State Generate $\mu_n|V_n, \vec{\theta}_{n-1} \sim N\left(\frac{a_3V_n+ b_3\sum_{i=1}^I\theta_{n-1,i}}{V_n+Ib_3},\frac{V_nb_3}{V_n+Ib_3}\right)$
		\State For each $i\in I$, generate $\theta_{n,i}|V_n, W_n, \mu_n \sim N\left(\frac{\mu_n W_n+ V_n\sum_{j=1}^JY_{ij}}{W_n+JV_n},\frac{V_nW_n}{W_n+JV_n}\right)$
		\EndFor
	\end{algorithmic}
\end{algorithm}
%

Suppose that we have the following proposal density,
\begin{align}\label{eqn:nu}
	f_{\nu}(\vec{\theta}, V, W, \mu) =
	& f_{IG}(V \mid 2a_1+1, 2b_1-1) f_{IG}(W \mid 2a_2+1, 2b_2) f_{N}(\mu \mid a_3, b_3) \prod_{i=1}^I f_{N}\left(\theta_i \mid \bar{Y}_i,\frac{W}{2J_i}\right) 
\end{align}

The following lemma bounds the Wasserstein distance between the variance component model Gibbs sampler and the stationary distribution using the above proposal distribution.

\begin{lemma}\label{lem:regKbound2}    
	Let $X_n = (\vec{\theta}_n, V_n, W_n, \mu_n)$ and $X'_n = (\vec{\theta}'_n, V'_n, W'_n, \mu'_n)$ be two copies of the variance component Gibbs sampler initialized 
	with $X'_0 \sim \nu$ where $f_{\nu}$ is defined in Equation \ref{eqn:nu} and $X_0\indep X'_0$. Fix $p\in[1,\infty)$ 
	and let $\pi$ be the corresponding stationary distribution. 
	Then for $r\geq p, b_1>1/2$,
	\begin{align}\label{eq:varcompbnd}
		W_{d,p}(\mathcal{L}(X_n), 
		\pi)
		\leq K_{2}(\pi,\nu)^{1/r}E[d(X_n,X'_n)^{2r}]^{\frac{1}{2r}}
	\end{align}
	
	Where 
	\begin{itemize}
		\item $K_{2}(\pi,\nu)=\frac{1}{L}C_1 C_2^{1/2}$
		\item $L=\int_{B} (b_1^*)^{-a_1^*}
		e^{-\frac{(\mu-a_3)^2}{2b_3}}\left(b_2^*\right)^{-(a_2^*)} \mathrm{d}(\vec{\theta},\mu)$ where $B\subset \mathbb{R}^{I+1}$
		\item $C_1 = \frac{\Gamma(a_1)}{\Gamma(a_1^*) b_1^{a_1}} \frac{\Gamma(a_2)}{\Gamma(a_2^*)b_2^{a_2}}
		(2\pi)^{\sum_{i=1}^I J_i/2+(I+1)/2}b_3^{1/2}$
		\item $C_2 = \frac{b_1^{2a_1}}{\Gamma(a_1)^2} \frac{\Gamma(2a_1+1)}{(2b_1-1)^{2a_1+1}}
		\frac{b_2^{2a_2}}{\Gamma(a_2)^2} 
		\frac{\Gamma(2a_2 +1)}{(2b_2)^{2a_2 +1}}  \frac{\Gamma(I/2-1)}{(\prod_{i=1}^I J_i)^{1/2}2^{I+\sum_{i=1}^I J_i}\pi^{T+I/2}} \frac{\Gamma(T-1)}{S^{T-1}} $
		\item $	a_1^*=a_1+I/2$ and  $b_1^*=b_1+\frac{\sum_{i=1}^I(\theta_i-\mu)^2}{2}$  \item $a_2^*=a_2+\sum_{i=1}^I J_i/2$ and $b_2^*=b_2+\sum_{i=1}^{I} \sum_{j=1}^{J_i} \frac{(Y_{ij}-\theta_i)^2}{2}$ 
		\item $S=\sum_{i=1}^I (\sum_{j=1}^{J_i}Y_{ij}^2-J_i (\bar{Y}_i)^2)$ and $T=\sum_{i=1}^IJ_i-I/2$
	\end{itemize}
	The proof is in section \ref{pf:regKbound2}
\end{lemma}

Note that in Lemma \ref{lem:regKbound2} the value of $K_{2}(\pi,\nu)$ still has an integral, which must be bounded from below, $L=\int_B g(\vec{\theta}, v, w, \mu)\mathrm{d}(\vec{\theta}, V, W, \mu)$. To address this, we use the \texttt{adaptIntegrate} function in \texttt{R}.

For this variance component model, upper bounds in total variation distance have been established.  In Real Data Example \ref{numex:varcomp2}, we compare our bounds in Wasserstein distance to the total variation bounds generated in \cite{Jones2004SufficientBF}. 

\cite{varcomp} also provides an upper bound on the total variation distance between the variance component Gibbs sampler model and its corresponding stationary distribution (See Theorem 1 of \cite{varcomp}). However the constants required to obtain an explicit upper bound are difficult to calculate (see Remark 6 of \cite{varcomp}). An upper bound in total variation distance for this example is also established in \cite{simconv}. 
Discussions of convergence diagnostics on total variation distance for variations of this example can be found in \cite{alicia, varcompgeom, roman, tan, varcompgeyer}. 

Upper bounds in Wasserstein distance of the variance component model have also already been discussed. \cite{wasmeth} analyzes the variance component model when the prior on $\mu$ is $\mu \propto 1$ instead of $\mu\sim N(a_3,b_3)$. In particular Proposition 25 of \cite{wasmeth} states that if $\lim_{I\to\infty}J^2/(I^{3+\delta})=\infty$ for some $\delta>0$ and the model is properly specified (condition (E2) in \cite{wasmeth}), then the geometric rate of convergence for this model goes to $0$. Proposition 24 of \cite{wasmeth} also bounds the total variation distance in terms of the Wasserstein distance. Discussions of convergence rates for the variance component model in Wasserstein distance are also presented in \cite{davis,yang}.

\begin{numexmp1}\label{numex:varcomp2}

We are going to compare our convergence bounds to the bounds from \cite{Jones2004SufficientBF}. The data for this example is simulated in \cite{Jones2004SufficientBF} and summarised in Table \ref{table:varcompdata}. \cite{Jones2004SufficientBF} generates an upper bound on the total variation distance between a Markov chain and its corresponding stationary distribution using a DnM bound stated in \cite{rose} (see Theorem 12 of \cite{rose} and Theorem 3.1 of \cite{Jones2004SufficientBF}). The hyperparameters for the DnM bound are defined in Table \ref{table:varcompnum1}. 

\begin{table}
	\centering
	\begin{tabular}{|c|c|c|c|c|c|}
		\hline
		\textbf{Cell} & \textbf{1} & \textbf{2} & \textbf{3} & \textbf{4} & \textbf{5} \\
		\hline
		\textbf{$\bar{Y}_i$} & -0.80247 & -1.0014 & -0.69090 & -1.1413 & -1.0125 \\
		\hline
	\end{tabular}
	\caption{Summary of simulated data used in \cite{Jones2004SufficientBF} (see their Table 1) and in Numerical Example \ref{numex:varcomp2}. $I=5$, $J=10$, $\sum_{i=1}^5\sum_{j=1}^{10}(y_{ij}-\bar{y}_i)^2=32.990$}
	\label{table:varcompdata}
\end{table}

\begin{table}
	\centering
\begin{tabular}{|c|c|c|c|c|c|c|}
	\hline
	\textbf{Hyperparameter} & \textbf{$a_1$} & \textbf{$b_1$} & \textbf{$a_2$} & \textbf{$b_2$} & \textbf{$a_3$} & \textbf{$b_3$} \\
	\hline
	\textbf{From} \cite{Jones2004SufficientBF} \textbf{and this paper} & 2.5 & 1 & 1 & 1 & $\bar{Y}$ & 1 \\
	\hline
\end{tabular}
\caption{Hyperparameters used in Numerical Example \ref{numex:varcomp2} and Table 2 of \cite{Jones2004SufficientBF}. 
	 $\bar{Y}=\frac{1}{IJ}\sum_{i=1}^I\sum_{j=1}^J Y_{ij}$ }
\label{table:varcompnum1}
\end{table}

%
%

Applying Lemma \ref{lem:regKbound2}, $K_{2}(\pi,\nu)= 48.7818$ and the following holds for $r= 1$, 

$$W_{||\cdot||_1,1}(\mathcal{L}(X_n),\pi)\leq 48.7818E\left[(||X_{n}-X'_{n}||_1)^{2}\right]^{\frac{1}{2}}$$


The initial value of $X_0=(\vec{\theta}_0,V_0,W_0,\mu_0)$ is $V_0\sim IG(a_1,b_1)$, $W_0\sim IG(a_2,b_2)$ and $(\vec{\theta}_0,\mu_0)=$  the initial value defined in \cite{Jones2004SufficientBF}, Remark 4.1 of \cite{Jones2004SufficientBF}. The initial value of $X'_0$ satisfies Equation \ref{eqn:nu}. 

We simulated $(||X_{n}-X'_{n}||_1)^2$ 1000 times ($M=1000$) using CRN. Figure \ref{fig:vcmcrn2} graphs the estimated upper bound on the Wasserstein distance between the Markov chain and its corresponding stationary distribution while Table \ref{table:civarcompgalinjones} provides the bound estimates and the $99\%$ confidence intervals using Lemma \ref{lem:ci}.

\begin{figure}
	\centering
	\title{Upper bound on $W_{||\cdot||_1,1}(\mathcal{L}(X_n),\pi)$}
	\includegraphics[width=0.7\linewidth]{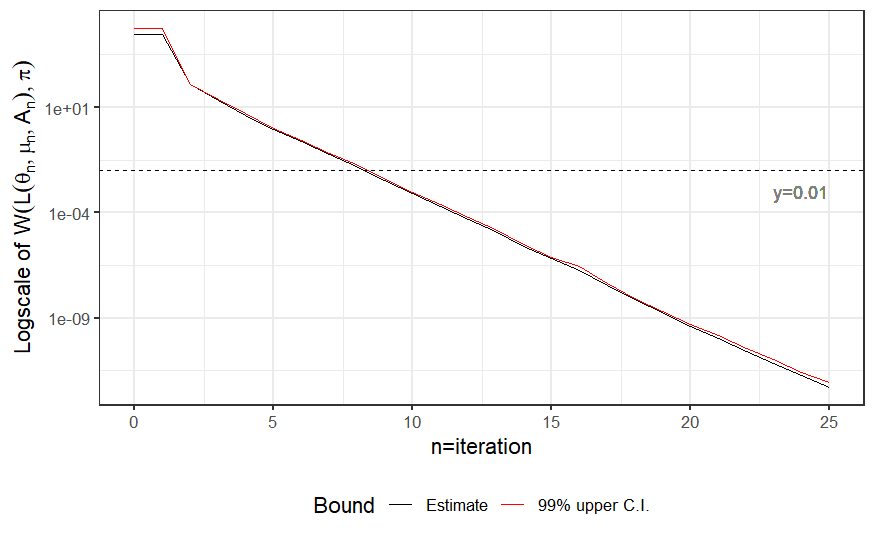}
	\caption{A Wasserstein distance upper bound on the Gibbs sampler for a variance component model (Numerical Example \ref{numex:varcomp2}) using CRN. The estimate uses the mean of 1000 simulations with initial values of $\mu_0,\vec{\theta}_0$ 
			satisfying the initial values defined in \cite{Jones2004SufficientBF} (see their Remark 4.1). The vertical axis is a log scale.}
	\label{fig:vcmcrn2}
\end{figure}

	\begin{table}
	\centering
\begin{tabular}{|c|c|c|c|c|}
	\hline
	\textbf{$n$} & $\frac{1}{1000}\sum_{m=1}^{1000}(||X_{n,m}-X'_{n,m}||_1)^{2}$ & \textbf{$99\%$ lower C.I.} & \textbf{Estimated bound} & \textbf{$99\%$ upper C.I.} \\
	\hline
1 & 397506.3 & 0 & 3075.601 & 55941.57 \\
\hline
2 & 6.0002 & 112.7212 & 119.4924 & 125.8998 \\
\hline
3 & 0.1835 & 17.3169 & 20.8980 & 23.9495 \\
\hline
4 & 0.0076 & 3.2230 & 4.2734 & 5.1123 \\
\hline
5 & 0.0004 & 0.7855 & 0.9132 & 1.0251 \\
\hline
6 & 2.4539e-05 & 0.1837 & 0.2417 & 0.2881 \\
\hline
7 & 1.4015e-06 & 0.0441 & 0.0578 & 0.0687 \\
\hline
8 & 1.0485e-07 & 0.0114 & 0.0158 & 0.0192 \\
\hline
9 & 5.4138e-09 & 0.0027 & 0.0036 & 0.0043 \\
	\hline
\end{tabular}
	\caption{Wasserstein bound estimates and confidence intervals for the variance component model, Numerical Example \ref{numex:varcomp2}.}
	\label{table:civarcompgalinjones}
\end{table}

The CRN bound indicates Wasserstein distance of less than 0.0099 by iteration 9 ($W_{||\cdot||_1,1}(\mathcal{L}(X_{9}),\pi)\\ \leq 0.00359$). The bound for \cite{Jones2004SufficientBF} indicates total variation distance of less than 0.00999 by iteration $3415$ ($||\mathcal{L}(X_{3415})-\pi||_{TV}\leq 0.00999$) (See Table 3 of \cite{Jones2004SufficientBF}). Note that the CRN bound is measured in Wasserstein distance while the bound in \cite{Jones2004SufficientBF} is measured in total variation distance.
The large disparity between the DnM and the CRN bound is consistent with the conclusions of \cite{dmlim} that say that the convergence rate
bounds based on single-step DnM tend to be overly conservative.

Finally we compare our convergence rates with informal methods in Section \ref{subsec:varcomp2}. The traceplot (Figure \ref{fig:traceplotvarcompJonesHobert}) suggests that convergence of $\mu$, $V$, and $W$ is nearly immediate. The Gelman-Rubin statistic is evaluated at different iterations in Table \ref{table:GRvarcompGH}. By iteration 10, the Gelman-Rubin statistic is close to or less than 1.05 for $V,W$, and $\mu$. 

\end{numexmp1}

\begin{numexmp}\label{numex:varcomp3}

We are given the yield of dyestuff from 6 batches of 5 yield measurements in grams (ie. $I=6$ and $J=5$). This data was published in Davies \cite{dyestuff}. The data can be structured as a variance component model where $\theta_i$ represents the mean yield for batch $i$.

We are going to compare our Wasserstein bound to the total variation bound from \cite{simconv}. \cite{simconv} generates an upper bound on the total variation distance between a Markov chain and its corresponding stationary distribution using a modified version of the popular DnM condition stated in \cite{rose} (see Theorem 12 of \cite{rose} and Proposition 1 of \cite{simconv}) and estimates the required DnM parameters ($\Lambda, \lambda, \epsilon$) using simulation. The hyperparameters for the DnM bound are defined in Table \ref{table:varcompnum3}. 

\begin{table}
	\centering
	\begin{tabular}{|c|c|c|c|c|c|c|}
		\hline
		\textbf{Hyperparameter} & \textbf{$a_1$} & \textbf{$b_1$} & \textbf{$a_2$} & \textbf{$b_2$} & \textbf{$a_3$} & \textbf{$b_3$} \\
		\hline
		\textbf{From} \cite{simconv} & 0.5 & 1 & 1 & 0 & 0 & $10^{12}$ \\
		\hline
		\textbf{This paper} & 0.5 & 1 & 1 & 1 & 0 & $10^{12}$ \\
		\hline
	\end{tabular}
	\caption{Hyperparameters used in Real Data Example \ref{numex:varcomp3} and \cite{simconv} (See their Section 4). The hyperparameters follow the same notation as in Equation \ref{eq:par}. }
	\label{table:varcompnum3}
\end{table}
 
In our comparison, we will first generate an upper bound on the Wasserstein distance between the Markov chain and the corresponding stationary distribution. Next, we will generate an upper bound on the total variation distance based on the Wasserstein distance using Proposition 24 of \cite{wasmeth}.

Proposition 24 of \cite{wasmeth} will be used to bound the total variation distance given the Wasserstein distance. Note that in \cite{wasmeth}, the prior on $\mu\sim 1$ is flat. Since in this example $\mu\sim N(0,10^{12})$, we assume that they are the same.

\begin{prop}[Proposition 24 of \cite{wasmeth}]\label{prop:wasmeth}
	Let $X_n$ be a variance component Gibbs sampler generated according to the priors of $V\sim IG(a_1,b_1), W\sim IG(a_2, b_2)$, and $\mu\propto 1$. Assume that $I\geq 2$. Then for each $n\geq 1$ and $X_0=(\vec{\theta_0},V_0,W_0, \mu_0)\in \mathbb{R}^{I}\times \mathbb{R}_+\times \mathbb{R}_+\times \mathbb{R}$,
	
	\begin{equation}
		||\mathcal{L}(X_n)-\pi||_{TV} \leq K_{TV} W_{||\cdot||_1,1}(\mathcal{L}(X_{n-1}),\pi)
	\end{equation}
	where $K_{TV}=(c_1+c_2) J^{3/2}I$ and 
	\begin{itemize}
		\item $c_1=\frac{2}{I}\left(\frac{I}{2}+a_1\right)\left(1+\sqrt{\frac{2}{b_1}}\frac{1}{I}+\frac{1}{2b_1 I^2}\right)^{I/2+a_1-1}\left(\sqrt{\frac{2}{b_1}}+\frac{1}{b_1I}\right)$
		\item $c_2=\frac{2}{IJ^{3/2}}\left(\frac{IJ}{2}+a_2\right)\left(1+\frac{2}{\sqrt{b_2}IJ}+\frac{1}{b_2I^2J^2}\right)^{IJ/2+a_2-1}\left(2\sqrt{\frac{J}{b_2}}+\frac{2}{b_2\sqrt{J}I}\right)$
	\end{itemize}
\end{prop} 
We bound the whole chain in Proposition \ref{prop:wasmeth} $(\vec{\theta_n},V_n,W_n, \mu_n)$ rather than the marginal chain $(\vec{\theta_n}, \mu_n)$ as presented in Proposition 24 of \cite{wasmeth} due to de-initialization properties \cite{roberts:rosenthal:2001}. See section 5.1 of \cite{wasmeth} for more details.

Applying \texttt{R} simulation to Lemma \ref{lem:regKbound2}, we calculate $K_{2}(\pi,\nu)= 0.00198$. The following holds, 
$$W_{||\cdot||_1,1}(\mathcal{L}(X_n),\pi)\leq 0.00198 E\left[(||X_{n}-X'_{n}||_1)^{2} \right]^{\frac{1}{2}}$$

Applying Proposition \ref{prop:wasmeth}, $K_{TV}=611.4339$ and so,
$$||\mathcal{L}(X_n)-\pi||_{TV} \leq 1.21338 E\left[(||X_{n}-X'_{n}||_1)^{2} \right]^{\frac{1}{2}}$$

The initial value of $X_0$ is $\vec{\theta}_0=$ the initial values defined in Equation 9 of \cite{simconv}, $V_0\sim IG(a1,b1)$, $W_0\sim IG(a2,b2)$ and $\mu_0=\bar{Y}$. The initial value of $X'_0$ is distributed according to Equation \ref{eqn:nu}. 

We simulated $X_{n}$ and $X'_{n}$ 500 times ($M=500$) using CRN. Figure \ref{fig:vcmcrnvsdnm} graphs the estimated upper bound on the total variation distance between the Markov chain and its corresponding stationary distribution. Table \ref{table:civarcomprosecowles} provides the total variation bound estimates and the $99\%$ confidence intervals using Lemma \ref{lem:ci}.

The CRN bound indicates total variation distance of less than 0.01 by iteration 461 ($||\mathcal{L}(X_{461})-\pi||_{TV}\leq 0.0037$). The DnM bound indicates total variation distance of less than 0.01 by iteration $98,750$ ($||\mathcal{L}(X_{98,750})-\pi||_{TV}\leq 0.0072$). The two bounds are compared in Figure \ref{fig:vcmcrnvsdnm} and Table \ref{table:civarcomprosecowles}. 
The large disparity between the DnM and the CRN bound is consistent with the conclusions of \cite{dmlim} that say that the convergence rate
bounds based on single-step DnM tend to be overly conservative. See Figure \ref{fig:vcmcrnvsdnm} for a graphical comparison between the two bounds.

\begin{figure}
	\centering
	\title{Comparison of bounds on $||\mathcal{L}(X_n)-\pi||_{TV}$}
	\includegraphics[width=0.7\linewidth]{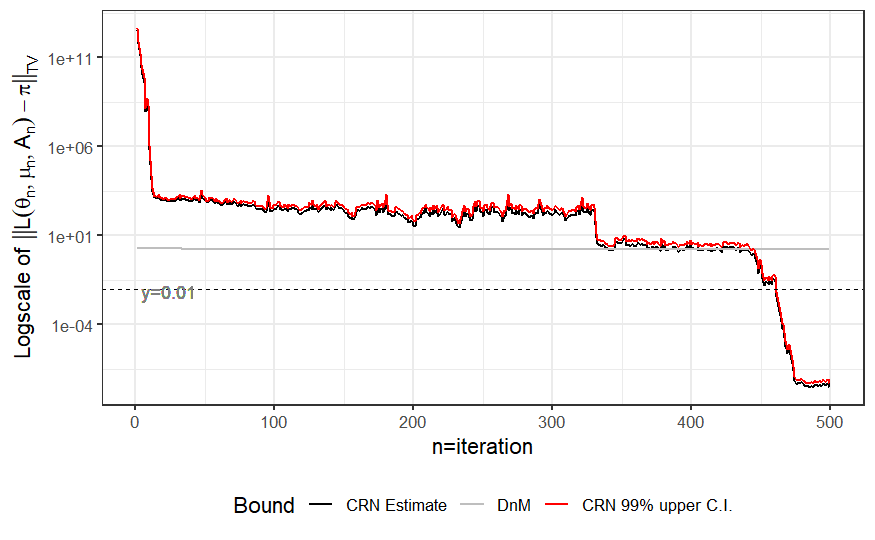}
	\caption{Comparison of total variation distance upper bounds on the Gibbs sampler for a variance component model (Real Data Example \ref{numex:varcomp3}) using CRN simulation and a DnM bound. The CRN estimate uses the mean of 500 simulations. The initial values of 
		$X_0$ satisfy Equation 9 in \cite{simconv}. }
	\label{fig:vcmcrnvsdnm}
\end{figure}

	\begin{table}
	\centering
\begin{tabular}{|c|c|c|c|c|c|}
	\hline
	$n$ & $\frac{1}{500}\sum_{m=1}^{500}(||X_{n,m}-X'_{n,m}||_1)^{2}$ & \textbf{$99\%$ lower C.I.} & \textbf{CRN} & \textbf{$99\%$ upper C.I.} & \textbf{DnM} \\
	\hline
	10 & 3.1983e+11 & 0 & 686202.5 &1210249 & 1.7888 \\
	\hline
	50 & 574231.3 & 413.0225 & 919.4749 & 1232.997 & 1.7821 \\
	\hline
	100 & 61787.44 & 0 & 301.6105 & 483.7103 & 1.7737 \\
	\hline
	250 & 10979.74 & 0 & 127.1430 & 231.2048 & 1.7490 \\
	\hline
	461 & 9.0992e-06 & 0 & \textbf{0.0037} & 0.0067 & 1.7149 \\
	\hline
\end{tabular}
	\caption{Total variation bound estimates and confidence intervals for the variance component model, Real Data Example \ref{numex:varcomp3}. Comparison of the estimated total variation bound using CRN with the total variation bound generated from \cite{simconv}, which was calculated using DnM.}
	\label{table:civarcomprosecowles}
\end{table}

Finally we compare our convergence rates with informal methods in Section \ref{subsec:varcomp3}. The traceplot (Figure \ref{fig:traceplotvarcompCowlesRose}) suggests that convergence of $\mu$, $V$, and $W$ to stationarity occurs after the first few iterations (the variance of $\mu$ seems to slightly increase with $n$, however). The Gelman-Rubin statistic is evaluated at different iterations in Table \ref{table:GRvarcompCowlesRose}. By iteration 250, the Gelman-Rubin statistic is close to or less than 1.05 for $V,W$, $\mu$, and $\theta_1$. 
\end{numexmp}

\section{Gibbs sampler for a Bayesian linear regression model }
\label{sec-example}
Consider a Gibbs sampler for the Bayesian linear regression model with semi-conjugate priors (see Chapter 5 of \cite{bayemethfin}) for which we will apply Theorem \ref{thm:summary} to simulate convergence bounds on the Wasserstein distance. For this particular example, we can also provide an upper bound on the total variation distance given the Wasserstein distance using similar calculations as \cite{sixta}. 

\begin{example}[Bayesian linear regression model]\label{ex:bayesianreggibbsindprior}
	Suppose that we have the following observed data $Y\in \mathbb{R}^k$ and $X\in \mathbb{R}^{k\times q}$ where
	$$Y | \beta, \sigma^2 \sim N_{k}(X\beta, \sigma^2I_k)$$
	for unknown parameters $\beta\in \mathbb{R}^q, \sigma^2\in \mathbb{R}_+$. Suppose that we apply the prior distributions on the unknown parameters,
	\begin{align*}
		&\beta \sim N_q(\beta_0, \Sigma_{\beta}) 
		&\sigma^2\sim IG\left(\frac{\upsilon_0}{2},\frac{\upsilon_0 c_0^2}{2}\right).
	\end{align*}
	The joint posterior density function of $\beta,\sigma^2|Y, X$ is proportional to the following equation,
	\begin{align}
		&g(\beta,\sigma^2) \\
		&=\frac{1}{(\sigma^2)^{(k+\upsilon_0)/2+1}}\exp\left(-\frac{1}{2\sigma^2} (Y-X\beta)^T(Y-X\beta) -\frac{1}{2}(\beta-\beta_0)^T\Sigma_{\beta}^{-1} (\beta-\beta_0) -\frac{\upsilon_0 c_0^2}{2\sigma^2}\right).
	\end{align}		
\end{example}

The Gibbs sampler on the Bayesian linear regression model is based on the conditional posterior distributions of $\beta_n, \sigma^2_n$ and is defined as follows initialized at $\beta_0, \sigma^2_0$:
\begin{enumerate}
	\item $\beta_n|\sigma_{n-1}^2,Y, X \sim N_q(\tilde{\beta}_{\sigma^2_{n-1}}, V_{\beta, \sigma^2_{n-1}})$
	\item $\sigma^2_n|\beta_n,Y, X \sim IG\left(\frac{k+\upsilon_0}{2},\frac{1}{2}\left[\upsilon_0 c_0^2+(Y-X\beta_n)^T(Y-X\beta_n)\right]\right)$
\end{enumerate}
where 
\begin{align*}
	&V_{\beta,\sigma^2_{n-1}}= \left(\frac{1}{\sigma^2_{n-1}}X^TX+\Sigma_{\beta}^{-1}\right)^{-1}, &\tilde{\beta}_{\sigma^2_{n-1}} = V_{\beta,\sigma^2_{n-1}}\left(\frac{1}{\sigma^2_{n-1}} X^T Y +\Sigma_{\beta}^{-1}\beta_0\right).
\end{align*}

Replacing $\beta_n=\tilde{\beta}_{\sigma^2_{n-1}} + V_{\beta, \sigma^2_{n-1}}^{1/2}Z_n$, where $Z_n\sim N_q(0,I_q)$ into the equation for $\sigma^2_n$ and with $IG_n\sim IG(\frac{k+\upsilon_0}{2},1)$, we get that
\begin{equation}\label{eqn:sigma_linereg}
	\sigma^2_n|\sigma^2_{n-1}, Y, X = \left[\frac{\upsilon_0 c_0^2}{2}+\frac{(X\tilde{\beta}_{\sigma^2_{n-1}} -Y +X V_{\beta,\sigma^{2}_{n-1}}^{1/2}Z_n)^T(X\tilde{\beta}_{\sigma^2_{n-1}} -Y +X V_{\beta,\sigma^{2}_{n-1}}^{1/2}Z_n)}{2}\right]IG_n
\end{equation} 
where $(Z_n, IG_n)_n$ are independent for all $n$.

Although the Markov chain may be high-dimensional in $\beta_n$, special properties of this Gibbs sampler allow us to upper bound the total variation distance between the law of two joint Markov chains, $\mathcal{L}(\sigma^2_n,\beta_n)$ and $\mathcal{L}(\sigma^{2}_{\infty},\beta_{\infty})=\pi$, in terms of the $L^1$ distance between $\sigma^2_n$ and $\sigma^{2}_{\infty}$ only.
\begin{lemma}\label{lem:regTVbound}
	The total variation distance between $\mathcal{L}(\beta_n,\sigma^2_n)$ and the corresponding stationary distribution, $\pi$, can be bounded by the expected distance as follows, 
	\begin{align*}
		\left\lVert \mathcal{L}(\beta_{n}, \sigma^{2}_{n})-\pi \right\rVert_{\text{TV}} &\leq \frac{(k + \upsilon_0)^2}{2 \upsilon_0 c_0^2} E[|\sigma^{2}_{n-1}-\sigma^{2}_\infty|]
	\end{align*}
	The proof is in section \ref{pf:regTVbound}.
\end{lemma}

Lemma \ref{lem:regTVbound} shows that total variation distance can ignore the difference between $\beta_n$ and $\beta_n'$. The Wasserstein distance, on the other hand, must take into account the distances between both $\beta_n, \beta_n'$ and $\sigma^2_n, \sigma^{2'}_n$ as shown below. 

\begin{lemma}\label{lem:regKbound}
	Let $(\beta_n, \sigma_n^2)_{n\geq 1}$ and $(\beta'_n, \sigma_n^{'2})_{n\geq 1}$ be two copies of the Markov chain with $\sigma^{'2}_0 \sim IG(\alpha', \beta')$ and  $\beta'_0 \sim N_q(\beta_0,\Sigma_{\beta})$ where $\alpha'=(k+\upsilon_0)/2-2$ and $\beta'=\upsilon_0 c_0^2/2$. Fix $p\in[1,\infty)$ and l
	et $\pi$ be the corresponding stationary distribution. 
	Then for $r>p$ 
	\begin{align*}
		W_{d,p}(\mathcal{L}(\beta_n, \sigma_n^2),\pi)
		\leq K_{1}(\pi,\nu)^{1/r}E[d((\beta_n, \sigma_n^2),(\beta'_n, \sigma_n'^2))^r]^{1/r}
	\end{align*}
	Where $K_{1}(\pi,\nu)=\frac{C}{L}$
	\begin{itemize}
		\item  $C=\frac{16 \Gamma((k+\upsilon_0)/2-2)(2\pi)^{q/2}\det(\Sigma_{\beta})^{1/2}}
		{\Gamma((k+\upsilon_0)/2)((Y-X\hat{\beta})^T(Y-X\hat{\beta}))^2 \left(\upsilon_0 c^2_0/2\right)^{(k+\upsilon_0)/2-2}e^{2}}$
		\item $L=\int_{B}
		\frac{e^{-\frac{1}{2}(\beta-\beta_0)^T\Sigma_{\beta}^{-1} (\beta-\beta_0)}}
		{(\upsilon_0 c_0^2/2 +	 (Y-X\beta)^T(Y-X\beta)/2)^{(k+\upsilon_0)/2}}
		\mathrm{d}(\beta)$ and 
		$B\subset \mathbb{R}^q\times \mathbb{R}_+$ is a bounded set.
	\end{itemize} 
	
	The proof is in Section \ref{pf:regKbound}
\end{lemma}

Convergence diagnostics for variations of the Bayesian linear regression Gibbs sampler can be found in \cite{sixta, raj, bayesregdootika, bayesregbacklund, lineregnongauss, bayesautoreg, linregbiswas}. However, we could not find any explicit bounds on this version of the Gibbs sampler for a Bayesian linear regression (ie. the bounds were applied to more simpler examples or the analysis did not provide an explicit bound). Using example \ref{ex:bayesianreggibbsindprior} notation, the papers \cite{sixta, raj, bayesregdootika} assume that $\Sigma_{\beta}=\sigma^2 \Sigma$ where $\Sigma$ is known. That is, the variance of $Y$ and the variance of $\beta$ are proportionally the same. Putting the results of \cite{sixta} and \cite{raj} together, it is shown that the geometric convergence rate that is generated when using one-shot coupling is the actual geometric convergence rate (see Theorem 3.1 of \cite{raj} compared with Theorem 4.4 of \cite{sixta}). The actual geometric convergence rate is calculated using the spectral gap in \cite{raj}. The papers (\cite{bayesregbacklund, lineregnongauss}) assume that the error of $Y|\beta,\sigma^2$ are non-Gaussian (to take into account outliers) and provide sufficient conditions for geometric ergodicity with respect to $k$ (dimension of $Y$) and $q$ (dimension of $\beta$). \cite{bayesautoreg} tells us that the Bayesian regression Gibbs sampler with semi-conjugate priors is geometrically ergodic regardless of the size of $k$ compared to $q$. Finally, \cite{linregbiswas} generates upper bounds in total variation between two copies of a linear regression Gibbs sampler model. The model in \cite{linregbiswas} adds slightly more complexity to example \ref{ex:bayesianreggibbsindprior} by taking $\Sigma_{\beta}$ as random. Similar to our results, the chains converge quickly.

Using Lemma \ref{lem:regKbound} we show how an upper bound on the convergence rate in Wasserstein distance can be simulated in Real Data Example \ref{ex:bayesianreggibbsindprior}. We further provide an estimate of the upper bound in total variation using Lemma \ref{lem:regTVbound}.

\begin{numexmp}\label{numex:linreg}
	Suppose that we are interested in evaluating the carbohydrate consumption (Y) by age, relative weight, and protein consumption (X) for twenty male insulin-dependent diabetics. For more information on this example, see Section 6.3.1 of \cite{glm}.
	
	We want to find the estimated upper bound on the Wasserstein and total variation distance to stationarity for a Bayesian linear regression Gibbs sampler. In this case, there are 20
	observed values and 4 parameters ($k = 20, q=4$). We set the priors to $\beta_0=\vec{0}$, $\Sigma_{\beta}=I_4$, $\upsilon_0=6$, $c_0^2 = 140$. The initial values of our Markov chains are $\sigma^{2}_0=100$, $\beta_0 \sim N_4(0, I_4)$, $\sigma^{'2}_0 \sim IG(11, 420)$ and  $\beta'_0 \sim N_4(0,I_4)$. We can apply Lemma \ref{lem:regKbound} as follows, where $K_{1}(\pi,\nu)=1 025 971$ and $r=5$.
	
	$$W_{||\cdot||_1,1}(\mathcal{L}(\sigma^2_n,\beta_n), \pi)\leq 15.93 E[ ||(\sigma^2_{n},\beta_{n})-(\sigma^{2'}_{n},\beta_{n})||_1^5]^{1/5}$$ holds almost surely.
	
	By Lemma \ref{lem:regTVbound} the total variation distance is bounded above by $\frac{(k+\upsilon_0)^2}{2\upsilon_0 c_0^2}=0.40238$ times the expected distance between $\sigma^2_n$ and $\sigma^{2}_{\infty}$ and so,
	\begin{align*}
		||\mathcal{L}(\sigma^{2}_n, \beta_n)-\pi||_{TV} &\leq 6.41 E\left[||(\sigma^2_{n},\beta_{n})-(\sigma^{2'}_{n},\beta_{n})||_1^5\right]^{1/5}
	\end{align*}
	
	Using CRN simulation, we simulated $||(\sigma^2_{n},\beta_{n})-(\sigma^{2'}_{n},\beta_{n})||_1$ ten thousand times ($M=10 000$) over 100 iterations ($N=100$). Figure \ref{fig:simdiff} graphs the upper bound on the total variation distance. 
	
	\begin{figure}
		\centering
		\title{Upper bound on $||\mathcal{L}(\sigma^2_n,\beta_n)- \pi||_{TV}$}
		\includegraphics[width=1\linewidth]{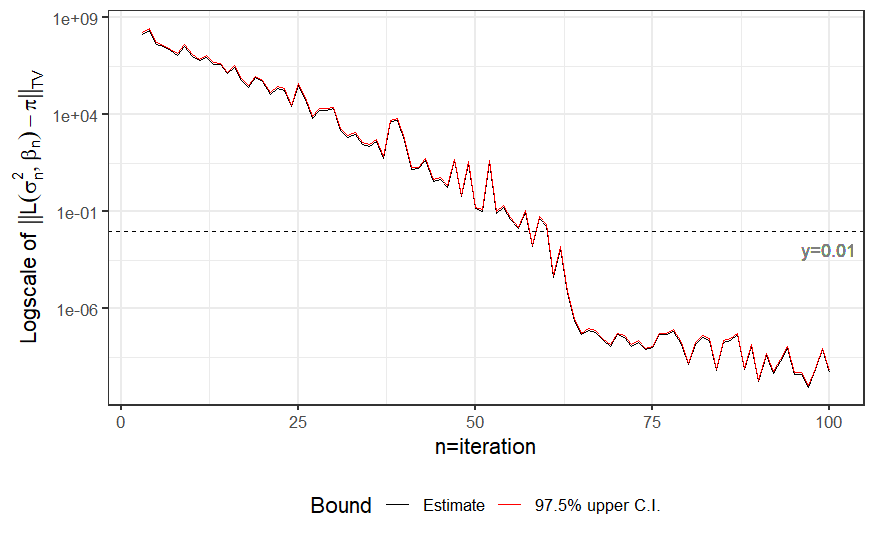}
		\caption{Total variation distance upper bound on the Gibbs sampler for a linear regression model (Real Data Example \ref{numex:linreg}) using CRN. The estimate uses the mean of 10000 simulations with initial values $\sigma^{2}_{0}=100$ and $\beta_0\sim N(0, I_q)$. The vertical axis is log-scaled and $r=5$. 
		}
		\label{fig:simdiff}
	\end{figure}
	
	By iteration 58, the total variation distance first hits below 0.01, $||\mathcal{L}(\sigma^2_{58},\beta_{58})- \pi||_{TV}\leq 0.0016$. In comparison, the Gelman-Rubin statistic first hits below 1.05 at iteration 49 (see Figure \ref{fig:linreg_gr}) and the traceplot suggests convergence after the first few iterations (see Figure \ref{fig:linreg_tracepot}).
	
\end{numexmp}

\section{Acknowledgement}

We thank the referees for their many excellent comments and suggestions. This research was partially funded by the National Science and Engineering Research Council of Canada (NSERC).

\section{Appendix}\label{sec:appendix}

\subsection{Convergence diagnostics}\label{subsec:convdiag}
\subsubsection{Real Data Example \ref{numex:stein}}\label{subsec:convdiagstein}

\begin{figure}[H]
	\centering
	\title{Gelman-Rubin statistic}
	\includegraphics[width=1\linewidth, height=7cm]{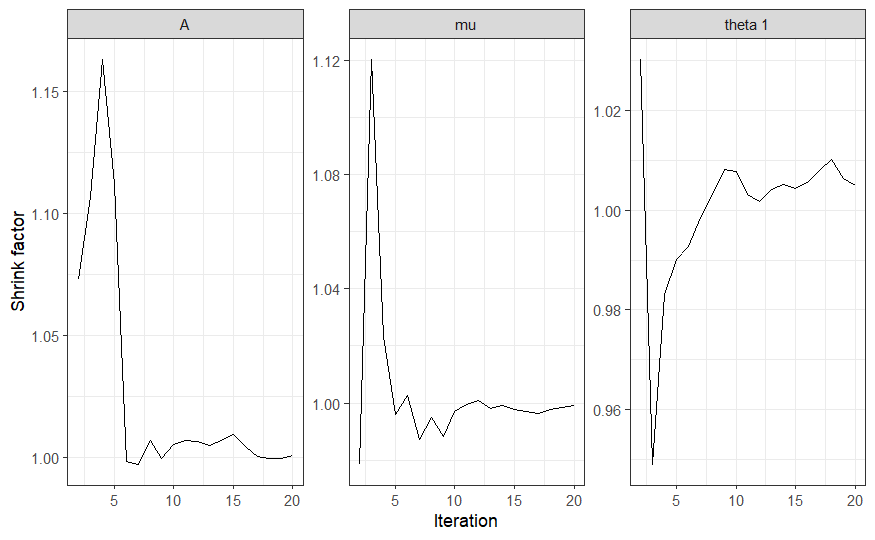}
	
	\title{Traceplot for iterations $1\leq n\leq 10$}
	\includegraphics[width=1\linewidth, height=7cm]{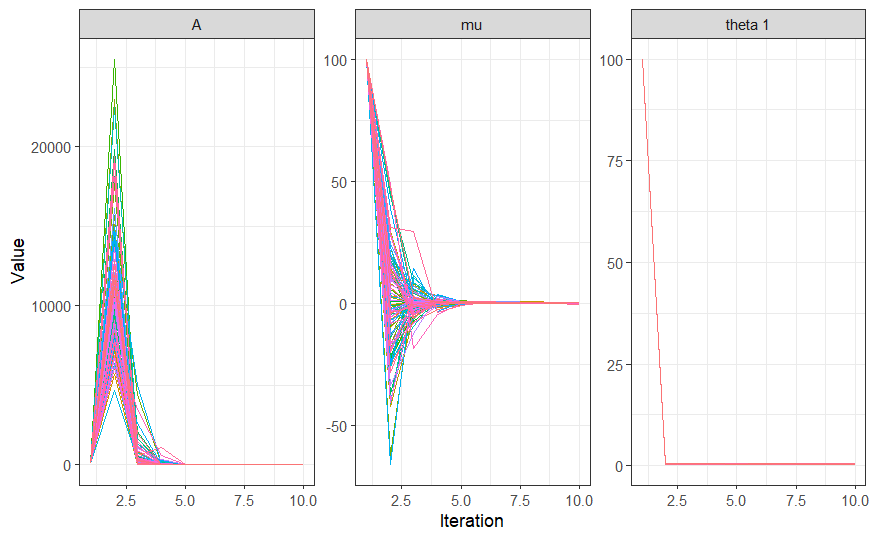}
	
	\caption{Gelman-Rubin statistic and traceplot for Real Data Example \ref{numex:stein}}
	\label{fig:traceplotstein}
\end{figure}

\subsubsection{Numerical Example \ref{numex:varcomp2}}\label{subsec:varcomp2}

\begin{table}[H]
	\centering
	\begin{longtable}{|c|c|c|c|}
		\hline
		\textbf{Iteration} & \textbf{V} & \textbf{W} & \textbf{$\mu$} \\
		\hline
		10 & 1.064556 & 1.046249 & 1.062693 \\
		\hline
		20 & 1.026606 & 1.021977 & 1.034042 \\
		\hline
		30 & 1.017409 & 1.013104 & 1.021825 \\
		\hline
		40 & 1.013610 & 1.011070 & 1.014756 \\
		\hline
		50 & 1.010736 & 1.008834 & 1.011396 \\
		\hline
		60 & 1.009296 & 1.007351 & 1.010150 \\
		\hline
	\end{longtable}
	\caption{Gelman-Rubin statistic at various iterations for Real Data Example \ref{numex:varcomp2}. The range of $\mu$ is -2 to 0.}
	\label{table:GRvarcompGH}
\end{table}

\begin{figure}[H]
	\centering
	\title{Traceplot}
	\includegraphics[width=1\linewidth, height=7cm]{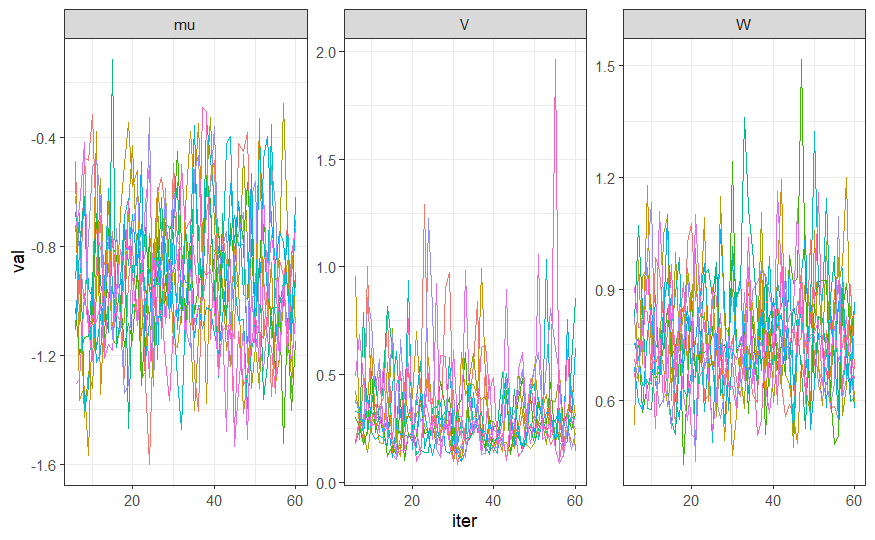}
	\caption{Traceplot for $\mu$, $V$, and $W$ for Real Data Example \ref{numex:varcomp2}}
	\label{fig:traceplotvarcompJonesHobert}
\end{figure}

\subsubsection{Real Data Example \ref{numex:varcomp3}}\label{subsec:varcomp3}

\begin{table}[H]
	\centering
\begin{longtable}{|c|c|c|c|c|}
	\hline
	$n$ & $V$ & $W$ & $\mu$ & $\theta_1$ \\
	\hline
	50 & 1.186270 & 2.052273 & 2.984139 & 2.792903 \\
	\hline
	100 & 1.095662 & 1.591515 & 1.597947 & 1.587649 \\
	\hline
	150 & 1.083923 & 1.454056 & 1.241374 & 1.215334 \\
	\hline
	200 & 1.062822 & 1.208700 & 1.086050 & 1.097018 \\
	\hline
	250 & 1.054723 & 1.074812 & 1.038007 & 1.056482 \\
	\hline
\end{longtable}
	\caption{Gelman-Rubin statistic at various iterations for Real Data Example \ref{numex:varcomp3}. The range of $V$ is $(0.5,3.5)$, $W$ is $1000, 20000$, and  $\mu,\theta_1$ is $(1400,1600)$. Note that because the range of initial values to calculate the Gelman-Rubin statistic is larger than the initial values in our bound, the convergence rates cannot be blindly compared. Rather, the Gelman-Rubin statistic indicates that the chains move relatively slowly.}
	\label{table:GRvarcompCowlesRose}
\end{table}

\begin{figure}[H]
	\centering
	\title{Traceplot}
	\includegraphics[width=1\linewidth, height=7cm]{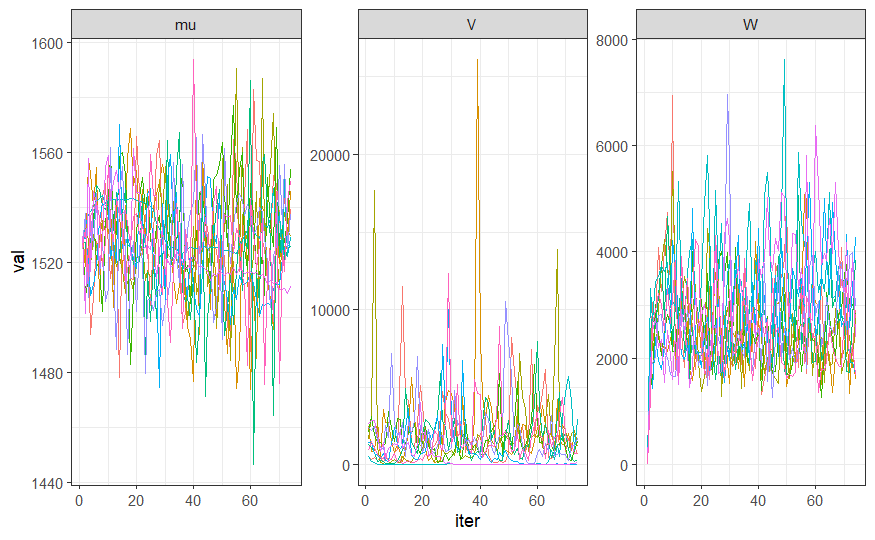}
	\caption{Traceplot for $\mu$, $V$, and $W$ for Real Data Example \ref{numex:varcomp3}}
	\label{fig:traceplotvarcompCowlesRose}
\end{figure}

\subsubsection{Real Data Example \ref{numex:linreg}}

\begin{figure}[H]
		\centering
		\title{Gelman-Rubin statistic}
		\includegraphics[width=1\linewidth, height=7cm]{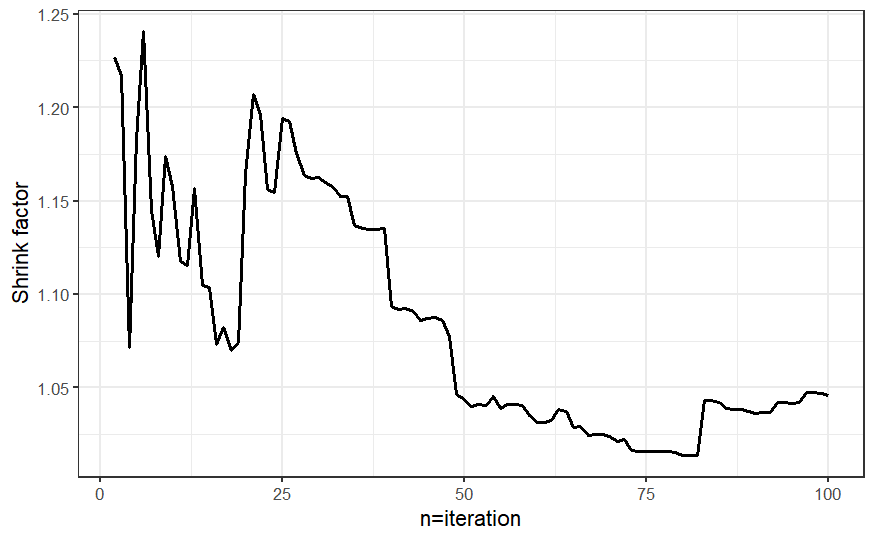}
		\caption{Gelman-Rubin statistic of $\sigma^2_n$.
		}
		\label{fig:linreg_gr}
\end{figure}
\begin{figure}[H]
		\includegraphics[width=1\linewidth, height=7cm] {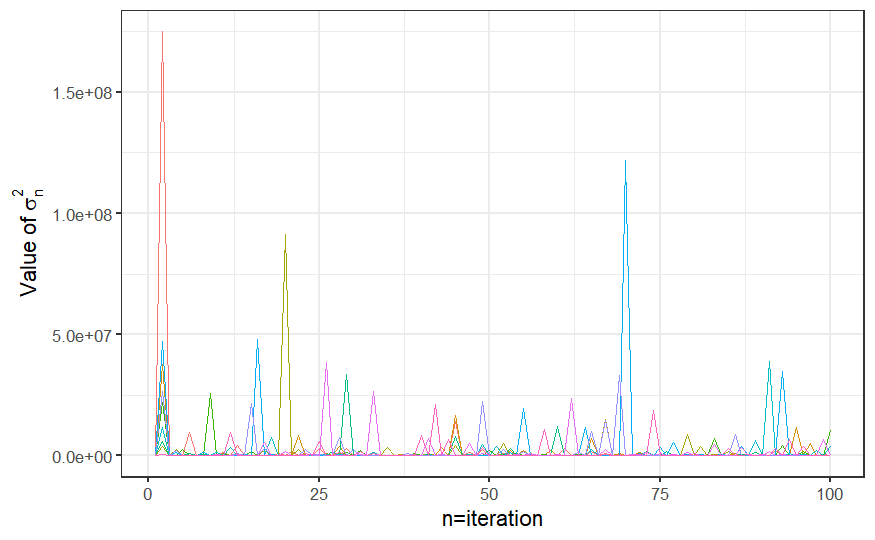}
		\caption{Traceplot of 10 simulations of $\sigma^2_n$ where 
			$\sigma^{2}_{0}=100$. 
		}
		\label{fig:linreg_tracepot}
\end{figure}

\subsection{Proof of Theorem \ref{thm:1}}

\begin{proof}[Proof of Theorem \ref{thm:1}]\label{pf:thm1original}
	
	Let $X_n, Y_n$ be two copies of the Markov chain such that $X_n=k_{\Theta_n}(X_0)$ and $Y_n=k_{\Theta'_n}(Y_0)$ (see Section \ref{sec:background} for more details on the notation). 
	If $X_{\infty}\sim \pi$, then 
	$k_{\Theta'_n}(X_{\infty})\sim \pi$. 
	
	First we show that $E[d(X_n,X_{\infty})^r]=E\left[d(X_n,Y_n)^r \frac{\mathrm{d}\pi(Y_0)}{\mathrm{d}\nu(Y_0)} \right]$ where $\mathrm{d}\pi(Y_0)/\mathrm{d}\nu(Y_0)=f_{\pi}(Y_0)/f_{\nu}(Y_0)$ and $\mathrm{d}\pi/\mathrm{d}\nu$ is the Radon-Nikodym derivative.
	\begin{align}\label{eqn:condchangevar}
		&E[d(X_n,X_{\infty})^r]\\
		&=E[E[d(k_{\Theta_n}(X_0),k_{\Theta'_n}(X_{\infty}))^r \mid X_0,\Theta_n, \Theta'_n]]\\
		&= E\left[\int_{\mathcal{X}} d(k_{\Theta_n}(X_0),k_{\Theta'_n}(y))^r \mathrm{d}\pi(y)\right] \\
		&= E\left[\int_{\mathcal{X}} d(k_{\Theta_n}(X_0),k_{\Theta'_n}(y))^r \frac{\mathrm{d}\pi(y)}{\mathrm{d}\nu(y)}\mathrm{d}\nu(y)\right] \\
		&=E\left[E\left[d(k_{\Theta_n}(X_0),k_{\Theta'_n}(Y_0))^r \frac{\mathrm{d}\pi(Y_0)}{\mathrm{d}\nu(Y_0)} \bigm\vert X_0,\Theta_n, \Theta'_n\right]\right]\\
		&=E\left[d(X_n,Y_n)^r \frac{\mathrm{d}\pi(Y_0)}{\mathrm{d}\nu(Y_0)} \right]
	\end{align}
	
	By Holder's inequality, for $s>1$ and $K_{s}(\pi,\nu)=E\left[\left(\frac{\mathrm{d}\pi(Y_0)}{\mathrm{d}\nu(Y_0)}\right)^{\frac{s}{s-1}}\right]^{\frac{s-1}{s}}$.
	\begin{align}\label{eqn:expchangevar}
		E[d(X_n,X_{\infty})^r]=E\left[d(X_n,Y_n)^r \frac{\mathrm{d}\pi(Y_0)}{\mathrm{d}\nu(Y_0)} \right]
		\leq K_{s}(\pi,\nu)E\left[d(X_n,Y_n)^{rs} \right]^{\frac{1}{s}}
	\end{align}
	Thus, the $L^r$-Wasserstein distance is bounded above as follows,
	\begin{align}\label{eq:thm1bullet2}
		W_{d,r}(\mathcal{L}(X_n),\pi)\leq E[d(X_n,X_{\infty})^r]^{\frac{1}{r}}
		\leq
		K_{s}(\pi,\nu)^{\frac{1}{r}} E\left[d(X_n,Y_n)^{rs} \right]^{\frac{1}{rs}}
	\end{align}
		Finally by Remark 6.6 of \cite{villani}, for $p\leq r$, $	W_{d,p}(\mathcal{L}(X_n),\pi)\leq 	W_{d,r}(\mathcal{L}(X_n),\pi)$ and so, by Equation \ref{eq:thm1bullet2},
	\begin{align*}
		W_{d,p}(\mathcal{L}(X_n),\pi)\leq W_{d,r}(\mathcal{L}(X_n),\pi)
		\leq 	K_{s}(\pi,\nu)^{\frac{1}{r}} E\left[d(X_n,Y_n)^{rs}\right]^{\frac{1}{rs}}
	\end{align*}
\end{proof}

\subsection{Alternative proof of Theorem \ref{thm:1} when $s=1$}

First we define essential supremum and infimum.

Let $f:\mathcal{X}\to \mathbb{R}$ be a function and $\lambda$ be the Lebesgue measure. The essential supremum is the smallest value $a\in \mathbb{R}$ such that $\lambda(x\mid f(x)<a)=1$. More formally, $ess \sup_{x} f(x)=\inf_{ a\in \mathbb{R} } \{ a \mid \lambda(x\mid f(x)>a)=  0 \}$. The essential infimum is likewise the largest value $a\in \mathbb{R}$ such that $\lambda(x\mid f(x)>a)=1$. Or, $ess \inf_{x} f(x)=\sup_{a\in \mathbb{R}} \{ a \mid \lambda(x\mid f(x)<a)=0\}$ \cite{esssup}.

	\subsubsection{Rejection rate and separation distance}
	
	 We will define rejection sampling and separation distance and show how they are related.
	
	\begin{definition}[Rejection sampling]
			Suppose that we have a target distribution $\pi$, which we want to sample from, but is difficult to do, and we have a proposal distribution $\nu$ that is easier to sample from. Suppose also that for the corresponding density functions $f_{\pi}(x)$ and $f_{\nu}(x)$,  $f_{\nu}(x)=0 \implies f_{\pi}(x)=0$, where $x\in \mathcal{X}$ and that $K_{1}(\pi,\nu) =ess \sup_{x\in\mathcal{X}} f_{\pi}(x)/f_{\nu}(x)$. To generate a random variable $X_{\infty}, \mathcal{L}(X_{\infty})=\pi$ we do the following,
			\begin{enumerate}
					\item Sample $X\sim \nu$ and $U\sim Unif(0,1)$ independently.
					\item If $U\leq\frac{1}{K_{1}(\pi,\nu)} \frac{f_{\pi}(X)}{f_{\nu}(X)}$ then accept $X$ as a draw from $\pi$. Otherwise reject $X$ and restart from step 1.
				\end{enumerate}
		\end{definition}

	\begin{definition}[Rejection sampler rejection rate]\label{def:K}
			For the rejection sampler algorithm defined above, denote the event $A=\{X\text{ is accepted as a draw from } \pi\}$. The rejection rate is $r(\pi, \nu)=1-P(A)=1-1/K_{1}(\pi,\nu)$ where $K_{1}(\pi,\nu)=ess \sup_{x\in \mathcal{X}} f_{\pi}(x)/f_{\nu}(x)$.
			\begin{proof}
					See Section 11.2.2 of \cite{ross}. 
				\end{proof}
		\end{definition}
	We further define the separation distance on the continuous state space $\mathcal{X}$ as follows. Separation distance was first defined in \cite{ALDOUS198769} for discrete state spaces. We use the definition of separation distance defined in \cite{sepdist} where the density functions are known.
	
	\begin{definition}[Separation distance (Remark 5 of \cite{sepdist})]
			Let $\nu$ and $\pi$ be two probability distributions defined on the same measurable space $(\mathcal{X},\mathcal{F})$ with densities $f_{\nu}$ and $f_{\pi}$ such that for $x\in \mathcal{X}$, $f_{\pi}(x)=0\implies f_{\nu}(x)=0$.
			The separation distance is
			$s(\pi, \nu)=ess \sup_{x}\left(1-\frac{f_{\nu}(x)}{f_{\pi}(x)}\right)$.
		\end{definition}
	It turns out that the separation distance and the rejection rate of the rejection sampler are the same.
	
	\begin{lemma}\label{lem:seprej}
			Let $\nu$ and $\pi$ be two distributions defined on the same measurable space $(\mathcal{X},\mathcal{F})$. If $f_{\nu}$ is the proposal density and $f_{\pi}$ is the target density in a rejection sampler, then the rejection rate equals the separation distance, $$s(\pi, \nu) = r(\pi, \nu)$$ 
		\end{lemma}
	\begin{proof}
			\begin{align*}
					s(\pi, \nu) &=ess \sup_{x\in \mathcal{X}}\left(1-\frac{f_{\nu}(x)}{f_{\pi}(x)}\right) = 1-ess \inf_{x\in \mathcal{X}}\frac{f_{\nu}(x)}{f_{\pi}(x)} \\
					&= 1-\frac{1}{ess \sup_{x\in \mathcal{X}}\frac{f_{\pi}(x)}{f_{\nu}(x)}} 
					= 1-\frac{1}{K_{1}(\pi,\nu)} 
					= r(\pi, \nu)
				\end{align*}
		\end{proof}
	
	
	\begin{proof}[Alternative proof of Theorem \ref{thm:1} for when $s_1=1,s_2=\infty$]\label{pf:thm1}
		In this proof, we use rejection sampling, which is similar to the proof method presented in \cite{crnconv}.
		Let $X_n, Y_n$ be two copies of the Markov chain and $A=\{\text{Accept $Y_0$ as a draw from }\pi\}$. Recall from Definition \ref{def:K} that $1-1/K_{1}(\pi,\nu)=r(\pi, \nu)=1-P(A)$ and so $P(A)=1/K_{1}(\pi,\nu)$. Note that $\mathcal{L}(Y_0|A)= \pi$ and so, $\mathcal{L}(Y_n|A)= \pi$.  
		\begin{align*}
			E[d(X_n,Y_n)^p] &=E[d(X_n,Y_n)^p\mid A]P(A) + E[d(X_n,Y_n)^p\mid A^c]P(A^c) \\
			&\geq E[d(X_n,X_{\infty})^p]P(A) \\
			&\geq E[d(X_n,X_{\infty})^p]\frac{1}{K_{1}(\pi,\nu)} 
		\end{align*}
		Thus,
		\begin{align*}
			W_{d,p}(\mathcal{L}(X_n),\pi)
			\leq E[d(X_n,X_{\infty})^p] \leq K_{1}(\pi,\nu) E[d(X_n,Y_n)^p]
		\end{align*}
	\end{proof}

Further note that by combining Definition \ref{def:K} and Lemma \ref{lem:seprej} we can write $K$ as follows 
$$K_{1}(\pi,\nu)= \frac{1}{1-s(\pi, \nu)} = \frac{1}{1-r(\pi, \nu)}$$

\subsection{Proof of lemma related to the Gibbs sampler for a James-Stein estimator}

\begin{proof}[Proof of Lemma \ref{lem:steinKbound}]\label{pf:steinKbound}
	First note that the density function of $\nu$ is as follows,
	\begin{align*}
		&f_{\nu}(\vec{\theta},\mu,A)\\ &=  \left(\prod_{i=1}^q\frac{1}{\sqrt{2\pi V}}e^{-(Y_i-\theta_i)^2/(2V)}  \right)  \left(\frac{\beta^{\alpha+(q-1)/2}}{\Gamma(\alpha+(q-1)/2)}A^{-\alpha-(q-1)/2-1}e^{-\beta/A}\right)
		\left(\frac{1}{\sqrt{2\pi A}}e^{-(\mu-\bar{\theta})/(2A)}\right)
	\end{align*} 
	
	Next, we want to calculate the supremum of the ratios where the function $g$ is defined as in Equation \ref{eqn:jsteinG}.
	\begin{align*}
		&\sup_{(\vec{\theta},\mu,A)\in \mathbb{R}^{q+1}\times \mathbb{R}_+}\frac{g(\vec{\theta},\mu,A)}{f_{\nu}(\vec{\theta},\mu,A)} \\
		&= \sup_{(\vec{\theta},\mu,A)\in \mathbb{R}^{q+1}\times \mathbb{R}_+}\frac{ \left(\prod_{i=1}^q\frac{1}{\sqrt{2\pi A}}e^{-(\theta_i-\mu)^2/(2A)} \right) \frac{\beta^{\alpha}}{\Gamma(\alpha)}A^{-\alpha-1}e^{-\beta/A}}{ \left(\frac{1}{\sqrt{2\pi A}}e^{-(\mu-\bar{\theta})^2/(2A)}\right) \frac{\beta^{\alpha+(q-1)/2}}{\Gamma(\alpha+(q-1)/2)}A^{-\alpha-q/2+1/2-1}e^{-\beta/A}} \\
		&=\frac{\Gamma(\alpha+(q-1)/2)}{\Gamma(\alpha)} \frac{1}{(2\pi\beta)^{(q-1)/2}} \sup_{(\vec{\theta},\mu,A)\in \mathbb{R}^{q+1}\times \mathbb{R}_+} e^{[(\bar{\theta}-\mu)^2-\sum_{i=1}^q(\theta_i-\mu)^2]/(2A)} \\
		&=\frac{\Gamma(\alpha+(q-1)/2)}{\Gamma(\alpha)(2\pi\beta)^{(q-1)/2}}
	\end{align*}
	
	
	The last equality follows since $(\bar{\theta}-\mu)^2\leq \sum_{i=1}^q(\theta_i-\mu)^2$ by Jensen's inequality. 
	
	Next we partially calculate $\int_{\mathbb{R}^{q+1}\times \mathbb{R}_+} g(\vec{\theta},\mu,A\mid Y) \mathrm{d}(\vec{\theta},\mu,A)$.
	\begin{align*}
		&\int_{\mathbb{R}^{q+1}\times \mathbb{R}_+} g(\vec{\theta},\mu,A\mid Y) \mathrm{d}(\vec{\theta},\mu,A)\\
		&= \int_{\mathbb{R}^{q+1}\times \mathbb{R}_+} \frac{\beta^{\alpha}}{\Gamma(\alpha)}A^{-\alpha-1}e^{-\beta/A}\frac{1}{(2\pi V)^{q/2}(2\pi A)^{q/2}}
		\left(\prod_{i=1}^qe^{-(\theta_i-Y_i)^2/(2V)-(\theta_i-\mu)^2/(2A)} \right)
		\mathrm{d}(\vec{\theta},\mu,A)  \\
		&= \int_{\mathbb{R}^{q+1}\times \mathbb{R}_+} \frac{\beta^{\alpha}}{\Gamma(\alpha)}A^{-\alpha-1}e^{-\beta/A}\frac{1}{(2\pi)^q (AV)^{q/2}}
		\left(\prod_{i=1}^qe^{-0.5[(\theta_i^2-2\theta_i Y_i+Y_i^2)/V+(\theta_i^2-2\theta_i\mu+\mu^2)/A]} \right)
		\mathrm{d}(\vec{\theta},\mu,A)  \\
		&= \int_{\mathbb{R}^{q+1}\times \mathbb{R}_+} \frac{\beta^{\alpha}}{\Gamma(\alpha)}A^{-\alpha-1}e^{-\beta/A}\frac{1}{(2\pi)^q (AV)^{q/2}}
		\left(\prod_{i=1}^qe^{-0.5\left[\theta_i^2(\frac{1}{V}+\frac{1}{A})-2\theta_i(\frac{Y_i}{V}+\frac{\mu}{A}) + (\frac{Y_i^2}{V}+\frac{\mu^2}{A}) \right]} \right)
		\mathrm{d}(\vec{\theta},\mu,A)  \\
		&= \int_{\mathbb{R}^{q+1}\times \mathbb{R}_+} \frac{\beta^{\alpha}}{\Gamma(\alpha)}A^{-\alpha-1}e^{-\beta/A}\frac{e^{-\frac{1}{2} \sum_{i=1}^q(\frac{Y_i^2}{V}+\frac{\mu^2}{A})} }{(2\pi)^q (AV)^{q/2}}
		\left(\prod_{i=1}^qe^{-0.5(\frac{1}{V}+\frac{1}{A})\left[\theta_i^2-2\theta_i\frac{\frac{Y_i}{V}+\frac{\mu}{A}}{\frac{1}{V}+\frac{1}{A}} \right]}
		\right)
		\mathrm{d}(\vec{\theta},\mu,A)  \\
		&= \int_{\mathbb{R}\times\mathbb{R}_+} \frac{\beta^{\alpha}}{\Gamma(\alpha)}A^{-\alpha-1}e^{-\beta/A}\frac{1}{(2\pi)^q (AV)^{q/2}}
		e^{-\frac{1}{2} \sum_{i=1}^q(\frac{Y_i^2}{V}+\frac{\mu^2}{A})} e^{\frac{1}{2}\sum_{i=1}^q\frac{(\frac{Y_i}{V}+\frac{\mu}{A})^2}{\frac{1}{V}+\frac{1}{A}}}
		\left(2\pi \frac{1}{(\frac{1}{V}+\frac{1}{A})}\right)^{q/2}
		\mathrm{d}(\mu,A)  \\
		&= \int_{\mathbb{R}\times\mathbb{R}_+} \frac{\beta^{\alpha}}{\Gamma(\alpha)}A^{-\alpha-1}e^{-\beta/A}\frac{1}{(2\pi)^{q} (AV)^{q/2}}
		e^{-\frac{1}{2} \sum_{i=1}^q(\frac{Y_i^2}{V}+\frac{\mu^2}{A})} e^{\frac{1}{2}\sum_{i=1}^q\frac{(\frac{Y_i}{V}+\frac{\mu}{A})^2}{\frac{1}{V}+\frac{1}{A}}}
		\left( 2\pi \frac{AV}{A+V}\right)^{q/2}
		\mathrm{d}(\mu,A)  \\
		&= \int_{\mathbb{R}\times\mathbb{R}_+} \frac{\beta^{\alpha}}{\Gamma(\alpha)}A^{-\alpha-1}e^{-\beta/A}\frac{1}{(2\pi)^{q/2}(A+V)^{q/2}}
		e^{-\frac{1}{2} \sum_{i=1}^q(\frac{Y_i^2}{V}+\frac{\mu^2}{A})} e^{\frac{1}{2}\sum_{i=1}^q\frac{(\frac{Y_i}{V}+\frac{\mu}{A})^2}{\frac{1}{V}+\frac{1}{A}}}
		\mathrm{d}(\mu,A)  \\
		&= \int_{\mathbb{R}\times\mathbb{R}_+} \frac{\beta^{\alpha}}{\Gamma(\alpha)}A^{-\alpha-1}e^{-\beta/A}\frac{1}{(2\pi)^{q/2}(A+V)^{q/2}}
		e^{\frac{1}{2}\sum_{i=1}^q\left[\frac{(\frac{Y_i}{V}+\frac{\mu}{A})^2}{\frac{1}{V}+\frac{1}{A}}-\frac{Y_i^2}{V}-\frac{\mu^2}{A}\right]} \mathrm{d}(\mu, A)
	\end{align*}

	Thus we get that 
	\begin{align*}
		&\frac{1}{\int_{\mathbb{R}^{q+1}\times \mathbb{R}_+} g(\vec{\theta},\mu,A\mid Y)}\mathrm{d}(\vec{\theta},\mu,A)\\
		&= \frac{\Gamma(\alpha)(2\pi)^{q/2}}{\beta^{\alpha}}
		\frac{1}{
			\int_{\mathbb{R}\times\mathbb{R}_+} \frac{A^{-\alpha-1}e^{-\beta/A}}{(A+V)^{q/2}}
			e^{\frac{1}{2}\sum_{i=1}^q\left[\frac{(\frac{Y_i}{V}+\frac{\mu}{A})^2}{\frac{1}{V}+\frac{1}{A}}-\frac{Y_i^2}{V}-\frac{\mu^2}{A}\right]}\mathrm{d}(\mu, A)}
	\end{align*}
	So 
	\begin{align*}
		K_{1}(\pi,\nu)&= \frac{1}{\int_{B}g(\vec{\theta},\mu,A) \mathrm{d}(\vec{\theta},\mu,A)} \sup_{(\vec{\theta},\mu,A)\in \mathbb{R}^{q+1}\times \mathbb{R}_+}\frac{g(\vec{\theta},\mu,A)}{f_{\nu}(\vec{\theta},\mu,A)} \\
		&\leq \frac{\Gamma(\alpha)(2\pi)^{q/2}}{\beta^{\alpha}}
		\frac{1}{
			\int_{\mathbb{R}\times\mathbb{R}_+} \frac{A^{-\alpha-1}e^{-\beta/A}}{(A+V)^{q/2}}
			e^{\frac{1}{2}\sum_{i=1}^q\left[\frac{(\frac{Y_i}{V}+\frac{\mu}{A})^2}{\frac{1}{V}+\frac{1}{A}}-\frac{Y_i^2}{V}-\frac{\mu^2}{A}\right]}\mathrm{d}(\mu, A)} \frac{\Gamma(\alpha+(q-1)/2)}{\Gamma(\alpha)(2\pi\beta)^{(q-1)/2}}\\
		&= \frac{\Gamma(\alpha+(q-1)/2)(2\pi)^{1/2}}{\beta^{\alpha+(q-1)/2}}
		\frac{1}{
			\int_{\mathbb{R}\times\mathbb{R}_+} \frac{A^{-\alpha-1}e^{-\beta/A}}{(A+V)^{q/2}}
			e^{\frac{1}{2}\sum_{i=1}^q\left[\frac{(\frac{Y_i}{V}+\frac{\mu}{A})^2}{\frac{1}{V}+\frac{1}{A}}-\frac{Y_i^2}{V}-\frac{\mu^2}{A}\right]}\mathrm{d}(\mu, A)} \\
	\end{align*}
	
	Let $L=	\int_{\mathbb{R}\times\mathbb{R}_+} \frac{A^{-\alpha-1}e^{-\beta/A}}{(A+V)^{q/2}}
	e^{\frac{1}{2}\sum_{i=1}^q\left[\frac{(\frac{Y_i}{V}+\frac{\mu}{A})^2}{\frac{1}{V}+\frac{1}{A}}-\frac{Y_i^2}{V}-\frac{\mu^2}{A}\right]}\mathrm{d}(\mu, A)$.
	Then by Equation \ref{eqn:kupper} coupled with Remark \ref{rem:inftynorm}, for $r\geq p$
	\begin{align*}
		W_{d,p}(\mathcal{L}(X_n),\pi)
		\leq \left(\frac{\Gamma(\alpha+(q-1)/2)(2\pi)^{1/2}}{\beta^{\alpha+(q-1)/2}}
		\frac{1}{L}
		E[d(X_n,Y_n)^r]\right)^{1/r}
	\end{align*}
\end{proof}
\newpage
\subsection{Proof of lemma related to the Gibbs sampler for a the variance component model}
\begin{proof}[Proof of Lemma \ref{lem:regKbound2}]\label{pf:regKbound2}
	
	First note that the unnormalized joint density function of $(\vec{\theta}, V,W, \mu)$ is written as follows ($f_{\pi}(\vec{\theta}, V,W, \mu)=\frac{1}{L}g(\vec{\theta}, V,W, \mu)$ where $L$ is unknown),
	\begin{align*}
		g(\vec{\theta}, V,W, \mu) &=\left(\frac{b_1^{a_1}}{\Gamma(a_1)}e^{-b_1/V}V^{-a_1-1}\right) \left(\frac{b_2^{a_2}}{\Gamma(a_2)}e^{-b_2/W}W^{-a_2-1}\right) 
		\left(\frac{1}{\sqrt{2\pi b_3}}e^{-\frac{(\mu-a_3)^2}{2b_3}}\right) \\
		&\left( \prod_{i=1}^{I}\frac{1}{\sqrt{2\pi V}}e^{-\frac{(\theta_i-\mu)^2}{2V}}\right)
		\left( \prod_{i=1}^{I} \prod_{j=1}^{J_i} \frac{1}{\sqrt{2\pi W}}e^{-\frac{(Y_{ij}-\theta_i)^2}{2W}}\right)\\
		&=f_{IG}(V\mid a_1, b_1)f_{IG}(W\mid a_2, b_2)f_{N}(\mu\mid a_3,b_3) \left(\prod_{i=1}^{I}f_N(\theta_i|\mu,V) \right) \\
		&\left( \prod_{i=1}^{I} \prod_{j=1}^{J_i} f_N(Y_{ij}|\theta_i, W) \right)
	\end{align*}
	
	For the variance component model, we will use Theorem \ref{thm:1} and set $s=2$ (so, $\frac{s}{s-1}=2$). Denote $L$ as a lower bound on $\int_{\mathbb{R}^I \times \mathbb{R}_+ \times \mathbb{R}_+ \times \mathbb{R}}g(\vec{\theta}, V,W, \mu)\mathrm{d}(\vec{\theta}, V, W, \mu)$. We want to find 
	\begin{align*}
		K_{s}(\pi,\nu)&= \left(\int_{\mathbb{R}^{I+1}\times \mathbb{R}_+^2}\left(\frac{f_{\pi}(\vec{\theta}, V, W, \mu)}{f_{\nu}(\vec{\theta}, V, W, \mu)}\right)^{2} f_{\nu}(\vec{\theta}, V, W, \mu) \mathrm{d}(\vec{\theta}, V, W, \mu)\right)^{1/2}\\
		&= \left(\int_{\mathbb{R}^{I+1}\times \mathbb{R}_+^2}\left(\frac{1}{L}\frac{g(\vec{\theta}, V, W, \mu)}{f_{\nu}(\vec{\theta}, V, W, \mu)}\right)^{2} f_{\nu}(\vec{\theta}, V, W, \mu) \mathrm{d}(\vec{\theta}, V, W, \mu)\right)^{1/2}\\
		&= \frac{1}{L} \left(\int_{\mathbb{R}^{I+1}\times \mathbb{R}_+^2}\frac{g(\vec{\theta}, V, W, \mu)^2}{f_{\nu}(\vec{\theta}, V, W, \mu)} \mathrm{d}(\vec{\theta}, V, W, \mu)\right)^{1/2}
	\end{align*}
	\textbf{Step 1: Find an upper bound on $\int_{\mathbb{R}^{I+1}\times \mathbb{R}_+^2}\frac{g(\vec{\theta}, V, W, \mu)^2}{f_{\nu}(\vec{\theta}, V, W, \mu)} \mathrm{d}(\vec{\theta}, V, W, \mu)$}
	
	To find an upper bound on the integral of the ratio, note the following ratios,
	
	\begin{align*}
		\frac{f_{IG}(V \mid a_1, b_1)^2}{f_{IG}(V \mid 2a_1+1, 2b_1-1)}&=\frac{\left(\frac{b_1^{a_1}}{\Gamma(a_1)}e^{-b_1/V}V^{-a_1-1}\right)^2}
		{\frac{(2b_1-1)^{2a_1+1}}{\Gamma(2a_1+1)}e^{-(2b_1-1)/V}V^{-(2a_1+1)-1}}\\
		&=\frac{\frac{b_1^{2a_1}}{\Gamma(a_1)^2}e^{-2b_1/V}V^{-2a_1-2}}
		{\frac{(2b_1-1)^{2a_1+1}}{\Gamma(2a_1+1)}e^{-(2b_1-1)/V}V^{-2a_1-2}}\\
		&=\frac{b_1^{2a_1}}{\Gamma(a_1)^2} \frac{\Gamma(2a_1+1)}{(2b_1-1)^{2a_1+1}}e^{-1/V}
	\end{align*}
	
	\begin{align*}
		\frac{f_{IG}(W \mid a_2, b_2)^2}{f_{IG}(W \mid 2a_2 +1, 2b_2)}&=\frac{\left(\frac{b_2^{a_2}}{\Gamma(a_2)}e^{-b_2/W}W^{-a_2-1}\right)^2}
		{\frac{(2b_2)^{2a_2 +1}}{\Gamma(2a_2 +1)}e^{-2b_2/W}W^{-2a_2-1-1}}\\
		&=\frac{\frac{b_2^{2a_2}}{\Gamma(a_2)^2}e^{-2b_2/W}W^{-2a_2-2}}{\frac{(2b_2)^{2a_2 +1}}{\Gamma(2a_2 +1)}e^{-2b_2/W}W^{-2a_2-2}}\\
		&=\frac{b_2^{2a_2}}{\Gamma(a_2)^2} 
		\frac{\Gamma(2a_2 +1)}{(2b_2)^{2a_2 +1}}
	\end{align*}
	
	
	\begin{align*}
		\frac{\prod_{i=1}^{I} \prod_{j=1}^{J_i} f_{N}(Y_{ij}\mid \theta_i, W)^2}{\prod_{i=1}^I f_{N}(\theta_i \mid \bar{Y}_i, \frac{W}{2J_i})}
		&=\frac{ \prod_{i=1}^{I}\prod_{j=1}^{J_i} \frac{1}{2\pi W}e^{-\frac{2(Y_{ij}-\theta_i)^2}{2W}}}{\prod_{i=1}^I \frac{1}{\sqrt{2\pi \frac{W}{2J_i}}} e^{-\frac{2J_i (\bar{Y}_i-\theta_i)^2}{2W}}}\\
		&=\frac{ \prod_{i=1}^{I} \frac{1}{(2\pi W)^{J_i}}e^{-\frac{\sum_{j=1}^{J_i}(Y_{ij}-\theta_i)^2}{W}}}
		{\prod_{i=1}^I \frac{1}{\sqrt{\pi W/J_i}} e^{-\frac{J_i (\bar{Y}_i-\theta_i)^2}{W}}}\\
		&=\prod_{i=1}^I \frac{\frac{1}{(2\pi W)^{J_i}}e^{-\frac{\sum_{j=1}^{J_i}Y_{ij}^2-2Y_{ij}\theta_i+\theta_i^2}{W}}}
		{\frac{1}{(\pi W/J_i)^{1/2}} e^{-\frac{J_i (\bar{Y}_i)^2-2J_i\bar{Y}_i\theta_i+J_i\theta_i^2}{W}}}\\
		&=\prod_{i=1}^I \frac{1}{(J_i)^{1/2}2^{J_i}(\pi W)^{J_i-1/2}} e^{\frac{J_i (\bar{Y}_i)^2-\sum_{j=1}^{J_i}Y_{ij}^2}{W}} \\
		&=\frac{1}{(\prod_{i=1}^I J_i)^{1/2}2^{\sum_{i=1}^I J_i}(\pi W)^{\sum_{i=1}^I J_i-I/2}} e^{\frac{\sum_{i=1}^I (J_i (\bar{Y}_i)^2-\sum_{j=1}^{J_i}Y_{ij}^2)}{W}}\\
		&=\frac{1}{(\prod_{i=1}^I J_i)^{1/2}2^{\sum_{i=1}^I J_i}(\pi W)^{T}} e^{\frac{-S}{W}}
	\end{align*}
	where $S=\sum_{i=1}^I (\sum_{j=1}^{J_i}Y_{ij}^2-J_i (\bar{Y}_i)^2)$ and $T=\sum_{i=1}^I J_i-I/2$. 
	
	We evaluate the following integral, which will be used evaluating the integral of $\frac{g(\vec{\theta}, V, W, \mu)^2}{f_{\nu}(\vec{\theta}, V, W, \mu)}$
	
	\begin{align*}
		\int_{\mathbb{R}^{I}}\prod_{i=1}^{I}f_N(\theta_i \mid \mu, V)^2\mathrm{d}\vec{\theta}&= \prod_{i=1}^{I} \int_{\mathbb{R}} \left(\frac{1}{\sqrt{2\pi V}}e^{-\frac{(\theta_i-\mu)^2}{2V}}\right)^2\mathrm{d}\theta_i \\
		&= \frac{1}{(2\pi V)^I} \prod_{i=1}^{I} \int_{\mathbb{R}} e^{-\frac{(\theta_i-\mu)^2}{V}}\mathrm{d}\theta_i \\
		&= \frac{1}{(2\pi V)^I} \prod_{i=1}^{I} \int_{\mathbb{R}} e^{-\frac{(\theta_i-\mu)^2}{2(V/2)}}\mathrm{d}\theta_i \\
		&= \frac{1}{(2\pi V)^I} \prod_{i=1}^{I} \sqrt{2\pi (V/2)} \\
		&= \frac{1}{(2\pi V)^I} (\pi V)^{I/2} \\
		&= \frac{1}{2^I(\pi V)^{I/2}} \\
	\end{align*}
	Next, we put all of the ratios together.
	
	We first integrate with respect to $\vec{\theta}$.
	
	\begin{align*}
		&\int_{\mathbb{R}^{I+1}\times \mathbb{R}_+^2}\frac{g(\vec{\theta}, V, W, \mu)^2}{f_{\nu}(\vec{\theta}, V, W, \mu)}\mathrm{d}(\vec{\theta}, V, W, \mu)\\
		&=\int_{\mathbb{R}^{I+1}\times \mathbb{R}_+^2}\frac{\left(f_{IG}(V\mid a_1, b_1) f_{IG}(W\mid a_2, b_2) f_N(\mu \mid a_3, b_3)
			\left( \prod_{i=1}^{I}f_N(\theta_i \mid \mu, V)\right)\right)^2
		}{f_{IG}(V \mid 2a_1+1, 2b_1-1) f_{IG}(W \mid 2a_2 +1, 2b_2) f_{N}(\mu \mid a_3,b_3) }\\
		&\hspace{2cm}\frac{ \prod_{i=1}^{I} \prod_{j=1}^{J_i} f_{N}(Y_{ij}\mid \theta_i, W)^2}{\prod_{i=1}^I f_{N}(\theta_i \mid \bar{Y}_i, W/2J_i)}
		\mathrm{d} (\vec{\theta}, V, W, \mu)\\
		&= \int_{\mathbb{R}^{I+1}\times \mathbb{R}_+^2} \frac{b_1^{2a_1}}{\Gamma(a_1)^2} \frac{\Gamma(2a_1+1)}{(2b_1-1)^{2a_1+1}}e^{-1/V} \frac{b_2^{2a_2}}{\Gamma(a_2)^2} 
		\frac{\Gamma(2a_2 +1)}{(2b_2)^{2a_2 +1}} f_N(\mu \mid a_3, b_3) \prod_{i=1}^{I}f_N(\theta_i \mid \mu, V)^2 \\
		&\hspace{2cm} \frac{1}{(\prod_{i=1}^I J_i)^{1/2}2^{\sum_{i=1}^I J_i}(\pi W)^{T}} e^{\frac{-S}{W}}\mathrm{d} (\vec{\theta}, V, W, \mu)\\
		&= \frac{b_1^{2a_1}}{\Gamma(a_1)^2} \frac{\Gamma(2a_1+1)}{(2b_1-1)^{2a_1+1}}
		\frac{b_2^{2a_2}}{\Gamma(a_2)^2} 
		\frac{\Gamma(2a_2 +1)}{(2b_2)^{2a_2 +1}} 
		\int_{\mathbb{R}^{1}\times \mathbb{R}_+^2} \left(\int_{\mathbb{R}^{I}}\prod_{i=1}^{I}f_N(\theta_i \mid \mu, V)^2 \mathrm{d}\vec{\theta} \right)e^{-1/V} f_N(\mu \mid a_3, b_3) \\
		&\hspace{2cm} \frac{1}{(\prod_{i=1}^I J_i)^{1/2}2^{\sum_{i=1}^I J_i}(\pi W)^{T}} e^{\frac{-S}{W}}\mathrm{d} (V, W, \mu)
		\end{align*}
		
		Next we integrate with respect to $V,W,\mu$
		\begin{align*}
		&= \frac{b_1^{2a_1}}{\Gamma(a_1)^2} \frac{\Gamma(2a_1+1)}{(2b_1-1)^{2a_1+1}}
		\frac{b_2^{2a_2}}{\Gamma(a_2)^2} 
		\frac{\Gamma(2a_2 +1)}{(2b_2)^{2a_2 +1}} 
		\int_{\mathbb{R}^{1}\times \mathbb{R}_+^2} \frac{1}{2^I(\pi V)^{I/2}} e^{-1/V} f_N(\mu \mid a_3, b_3) \\
		&\hspace{2cm} \frac{1}{(\prod_{i=1}^I J_i)^{1/2}2^{\sum_{i=1}^I J_i}(\pi W)^{T}} e^{\frac{-S}{W}}\mathrm{d} (V, W, \mu)\\
		&= \frac{b_1^{2a_1}}{\Gamma(a_1)^2} \frac{\Gamma(2a_1+1)}{(2b_1-1)^{2a_1+1}}
		\frac{b_2^{2a_2}}{\Gamma(a_2)^2} 
		\frac{\Gamma(2a_2 +1)}{(2b_2)^{2a_2 +1}}  \frac{1}{(\prod_{i=1}^I J_i)^{1/2}2^{I+\sum_{i=1}^I J_i}\pi^{T+I/2}} \\
		&\hspace{2cm} \int_{\mathbb{R}^{1}\times \mathbb{R}_+^2} \frac{1}{V^{I/2}} e^{-1/V} f_N(\mu \mid a_3, b_3) \frac{1}{W^{T}} e^{\frac{-S}{W}}\mathrm{d} (V, W, \mu)\\
		&= \frac{b_1^{2a_1}}{\Gamma(a_1)^2} \frac{\Gamma(2a_1+1)}{(2b_1-1)^{2a_1+1}}
		\frac{b_2^{2a_2}}{\Gamma(a_2)^2} 
		\frac{\Gamma(2a_2 +1)}{(2b_2)^{2a_2 +1}}  \frac{1}{(\prod_{i=1}^I J_i)^{1/2}2^{I+\sum_{i=1}^I J_i}\pi^{T+I/2}} \\
		&\hspace{2cm} \Gamma(I/2-1)\frac{\Gamma(T-1)}{S^{T-1}}\\
	\end{align*}

	To summarise, we get,
	\begin{align*}
		&C_2 = \int_{\mathbb{R}^{I+1}\times \mathbb{R}_+^2}\frac{g(\vec{\theta}, V, W, \mu)^2}{f_{\nu}(\vec{\theta}, V, W, \mu)}\mathrm{d} (\vec{\theta}, v, w, \mu)\\
	&=\frac{b_1^{2a_1}}{\Gamma(a_1)^2} \frac{\Gamma(2a_1+1)}{(2b_1-1)^{2a_1+1}}
	\frac{b_2^{2a_2}}{\Gamma(a_2)^2} 
	\frac{\Gamma(2a_2 +1)}{(2b_2)^{2a_2 +1}}  \frac{\Gamma(I/2-1)}{(\prod_{i=1}^I J_i)^{1/2}2^{I+\sum_{i=1}^I J_i}\pi^{T+I/2}} \frac{\Gamma(T-1)}{S^{T-1}} \\
	\end{align*}

	\textbf{Step 2: Find a lower bound on $\int_{\mathbb{R}^I \times \mathbb{R}_+ \times \mathbb{R}_+ \times \mathbb{R}}g(\vec{\theta}, v, w, \mu)\mathrm{d}(\vec{\theta},v, w, \mu)$}
	
	To find a lower bound over the integral we estimate the above integral over a closed set $B\subset \mathbb{R}^I \times \mathbb{R}_+ \times \mathbb{R}_+ \times \mathbb{R}$.

	\begin{align*}
		g(\vec{\theta}, V,W, \mu) &=\left(\frac{b_1^{a_1}}{\Gamma(a_1)}e^{-b_1/V}V^{-a_1-1}\right) \left(\frac{b_2^{a_2}}{\Gamma(a_2)}e^{-b_2/W}W^{-a_2-1}\right) 
		\left(\frac{1}{\sqrt{2\pi b_3}}e^{-\frac{(\mu-a_3)^2}{2b_3}}\right) \\
		&\left( \prod_{i=1}^{I}\frac{1}{\sqrt{2\pi V}}e^{-\frac{(\theta_i-\mu)^2}{2V}}\right)
		\left( \prod_{i=1}^{I} \prod_{j=1}^{J_i} \frac{1}{\sqrt{2\pi W}}e^{-\frac{(Y_{ij}-\theta_i)^2}{2W}}\right)\\
		&=\left(\frac{b_1^{a_1}}{\Gamma(a_1)}e^{-b_1/V}V^{-a_1-1}\right) \left(\frac{b_2^{a_2}}{\Gamma(a_2)}e^{-b_2/W}W^{-a_2-1}\right) 
		\left(\frac{1}{\sqrt{2\pi b_3}}e^{-\frac{(\mu-a_3)^2}{2b_3}}\right) \\
		&\left( \frac{1}{(2\pi V)^{I/2}}e^{-\frac{\sum_{i=1}^I(\theta_i-\mu)^2}{2V}}\right)
		\left(\frac{1}{(2\pi)^{\sum_{i=1}^I J_i/2} W^{\sum_{i=1}^I J_i/2}}e^{-\left(\sum_{i=1}^{I} \sum_{j=1}^{J_i} \frac{(Y_{ij}-\theta_i)^2}{2}\right)/W}\right)\\
		&=\left(\frac{b_1^{a_1}}{\Gamma(a_1)(2\pi)^{I/2}}e^{-\left(b_1+\frac{\sum_{i=1}^I(\theta_i-\mu)^2}{2}\right)/V}V^{-(a_1+I/2)-1}\right) \left(\frac{1}{\sqrt{2\pi b_3}}e^{-\frac{(\mu-a_3)^2}{2b_3}}\right)\\ &\left(\frac{b_2^{a_2}}{\Gamma(a_2)(2\pi)^{\sum_{i=1}^I J_i/2}}W^{-(a_2+\sum_{i=1}^I J_i/2)-1} e^{-\left(b_2+\sum_{i=1}^{I} \sum_{j=1}^{J_i} \frac{(Y_{ij}-\theta_i)^2}{2}\right)/W}\right)
	\end{align*}
	Define
	\begin{align*}
		&a_1^*=a_1+I/2 &b_1^*=b_1+\frac{\sum_{i=1}^I(\theta_i-\mu)^2}{2}\\
		&a_2^*=a_2+\sum_{i=1}^I J_i/2 
		&b_2^*=b_2+\sum_{i=1}^{I} \sum_{j=1}^{J_i} \frac{(Y_{ij}-\theta_i)^2}{2}
	\end{align*}
	
	Taking the integral of $g$ with respect to $V$ and $W$, we get,
	\begin{align*}
		&\int_{\mathbb{R}^I \times \mathbb{R}_+ \times \mathbb{R}_+ \times \mathbb{R}}g(\vec{\theta}, v, w, \mu)\mathrm{d}(\vec{\theta},v, w, \mu)\\
		&=\int_{\mathbb{R}^I \times \mathbb{R}_+ \times \mathbb{R}_+ \times \mathbb{R}} \left(\frac{b_1^{a_1}}{\Gamma(a_1)(2\pi)^{I/2}}e^{-b_1^*/V}V^{-a_1^*-1}\right) \left(\frac{1}{\sqrt{2\pi b_3}}e^{-\frac{(\mu-a_3)^2}{2b_3}}\right)\\ &\left(\frac{b_2^{a_2}}{\Gamma(a_2)(2\pi)^{\sum_{i=1}^I J_i/2}}W^{-a_2^*-1} e^{-b_2^*/W}\right) \mathrm{d}(\vec{\theta},v, w, \mu)\\
		&=\int_{\mathbb{R}^{I+1}} \left(\frac{\Gamma(a_1^*) b_1^{a_1}}{\Gamma(a_1)(2\pi)^{I/2}}\left(b_1^*\right)^{-a_1^*}\right) 
		\left(\frac{1}{\sqrt{2\pi b_3}}e^{-\frac{(\mu-a_3)^2}{2b_3}}\right)\\ 
		&\left(\frac{\Gamma(a_2^*)b_2^{a_2}}{\Gamma(a_2)(2\pi)^{\sum_{i=1}^I J_i/2}} \left(b_2^*\right)^{-a_2^*}\right) \mathrm{d}(\vec{\theta},\mu)
	\end{align*}
	
	The integral can then be defined as follows,
	\begin{align*}
		&\int_{\mathbb{R}^I \times \mathbb{R}_+ \times \mathbb{R}_+ \times \mathbb{R}}g(\vec{\theta}, v, w, \mu)\mathrm{d}(\vec{\theta},v, w, \mu)\\
		&=\frac{\Gamma(a_1^*) b_1^{a_1}}{\Gamma(a_1)} \frac{\Gamma(a_2^*)b_2^{a_2}}{\Gamma(a_2)}
		\frac{1}{(2\pi)^{\sum_{i=1}^I J_i/2+(I+1)/2}b_3^{1/2}} \int_{\mathbb{R}^{I+1}} (b_1^*)^{-a_1^*}
		e^{-\frac{(\mu-a_3)^2}{2b_3}}\left(b_2^*\right)^{-(a_2^*)} \mathrm{d}(\vec{\theta},\mu)\\
		&=C_1^{-1} \int_{\mathbb{R}^{I+1}} (b_1^*)^{-a_1^*}
		e^{-\frac{(\mu-a_3)^2}{2b_3}}\left(b_2^*\right)^{-(a_2^*)} \mathrm{d}(\vec{\theta},\mu)
	\end{align*}
	
	Where $C_1^{-1} = \frac{\Gamma(a_1^*) b_1^{a_1}}{\Gamma(a_1)} \frac{\Gamma(a_2^*)b_2^{a_2}}{\Gamma(a_2)}
	\frac{1}{(2\pi)^{\sum_{i=1}^I J_i/2+(I+1)/2}b_3^{1/2}}$
	
	Combining step 1 and step 2, we get,
	\begin{align*}
		&\frac{1}{\int_{B}g(\vec{\theta},V,W, \mu) \mathrm{d}(\vec{\theta},V,W, \mu)} \left(\int_{\mathbb{R}^{I+1}\times \mathbb{R}_+^2}\frac{g(\vec{\theta}, V, W, \mu)^2}{f_{\nu}(\vec{\theta}, V, W, \mu)}\mathrm{d}(\vec{\theta}, v, w, \mu)\right)^{1/2}\\
		&=\frac{1}{L}C_1 C_2^{1/2}
	\end{align*}
	
	Where $L=\int_{\mathbb{R}^{I+1}} (b_1^*)^{-a_1^*}
	e^{-\frac{(\mu-a_3)^2}{2b_3}}\left(b_2^*\right)^{-(a_2^*)} \mathrm{d}(\vec{\theta},\mu)$
	
	$C_1 = \frac{\Gamma(a_1)}{\Gamma(a_1^*) b_1^{a_1}} \frac{\Gamma(a_2)}{\Gamma(a_2^*)b_2^{a_2}}
	(2\pi)^{\sum_{i=1}^I J_i/2+(I+1)/2}b_3^{1/2}$
	
	$C_2 = \frac{b_1^{2a_1}}{\Gamma(a_1)^2} \frac{\Gamma(2a_1+1)}{(2b_1-1)^{2a_1+1}}
	\frac{b_2^{2a_2}}{\Gamma(a_2)^2} 
	\frac{\Gamma(2a_2 +1)}{(2b_2)^{2a_2 +1}}  \frac{\Gamma(I/2-1)}{(\prod_{i=1}^I J_i)^{1/2}2^{I+\sum_{i=1}^I J_i}\pi^{T+I/2}} \frac{\Gamma(T-1)}{S^{T-1}}  $
\end{proof}

\subsection{Proof of lemmas related to the Gibbs sampler of a  Bayesian linear regression model}

\begin{proof}[Proof of Lemma \ref{lem:regTVbound}]\label{pf:regTVbound}
	It follows from the de-initializing property of the Markov chain \cite[Example 3]{roberts:rosenthal:2001} that
	\[
	\left\lVert \mathcal{L}(\beta_{n + 1}, \sigma^{2}_{n + 1})-\mathcal{L}(\beta'_{n + 1}, \sigma^{'2}_{n + 1}) \right\rVert_{\text{TV}}
	\le \left\lVert \mathcal{L}(\sigma^{2}_{n})-\mathcal{L}(\sigma^{'2}_{n}) \right\rVert_{\text{TV}}.
	\]
	Using CRN, let $G_n=G'_n$ where $G_n \sim \text{Gamma}(\alpha,1)$ $Z_n\sim N(0, I_q)$ are independent and denote 
	\begin{align*}
		&W_n=\frac{\upsilon_0 c_0^2}{2}+\frac{ \left\lVert X\tilde{\beta}_{\sigma^2_{n-1}} -Y +X V_{\sigma^{2}_{n-1}}^{1/2}Z_n \right\rVert^2 }{2}
		\\
		&W'_n=\frac{\upsilon_0 c_0^2}{2}+\frac{ \left\lVert X\tilde{\beta}_{\sigma^{'2}_{n-1}} -Y +X V_{\sigma^{'2}_{n-1}}^{1/2}Z_n \right\rVert^2 }{2},
	\end{align*}
	so that $\sigma^2_n | \sigma^2_{n-1} = W_n \frac{1}{G_n}$ and $\sigma^{'2}_n | \sigma^{'2}_{n-1} =W'_n \frac{1}{G'_n}$. 
	Denote $\Delta=W'_n-W_n$ and without loss of generality, assume $W'_n>W_n$. 
	Since $G_n\sim \text{Gamma}(\alpha,1)$ where $\alpha=\frac{k+\upsilon_0}{2}$, let $\pi_{1/G_n}$ denote the density of $1/G_n$ and similarly denote the density $\pi_{(1+\Delta/W_n)/G_n}$ for $(1+\Delta/W_n)/G_n$.
	So we have
	\begin{align*}
		\pi_{1/G_n}(x) &\propto x^{-\alpha-1}e^{-1/x}\\
		\pi_{(1+\Delta/W_n)/G_n}(x) &\propto \frac{1}{1+\Delta/W_n} \left(\frac{x}{1+\Delta/W_n}\right)^{-\alpha-1}e^{-(1+\Delta/W_n)/x}\\ 
		&\propto (1+\Delta/W_n)^{\alpha}x^{-\alpha-1}e^{-(1+\Delta/W_n)/x}.
	\end{align*}
	Using the coupling characterization of total variation
	\begin{align*}
		&||\mathcal{L}(\sigma^{2}_n)-\mathcal{L}(\sigma^{'2}_n)||_{TV}\\ &\leq E\left[ ||\mathcal{L}(W_n/G_n \mid Z_n, \sigma_{n-1}^2, \sigma^{'2}_{n-1})-\mathcal{L}(W'_n /G_n\mid Z_n, \sigma_{n-1}^2, \sigma^{'2}_{n-1})||_{TV} \right]\\
		&\hspace{10cm} \text{By Proposition 2.2 of \cite{sixta}}\\
		&= E\left[ ||\mathcal{L}(W_n/G_n \mid Z_n, \sigma_{n-1}^2, \sigma^{'2}_{n-1})-\mathcal{L}((W_n +\Delta) /G_n\mid Z_n, \sigma_{n-1}^2, \sigma^{'2}_{n-1})||_{TV}  \right] &\\
		&= E\left[ ||\mathcal{L}\left(\frac{1}{G_n}\mid Z_n, \sigma_{n-1}^2, \sigma^{'2}_{n-1}\right)-\mathcal{L}\left(\left(1 +\frac{\Delta}{W_n}\right) \frac{1}{G_n} \mid Z_n, \sigma_{n-1}^2, \sigma^{'2}_{n-1}\right)||_{TV}  \right]\\
		&\hspace{10cm}\text{By Proposition 2.1 of \cite{sixta}}\\
		&\leq E\left[ E\left[\sup_{x>0} \left\{ 1- \frac{\pi_{1/G_n}(x)}{\pi_{(1+\Delta/W_n)/G_n}(x)} \right\} \bigm| Z_n, \sigma_{n-1}^2, \sigma^{'2}_{n-1} \right]\right]\\
		&\hspace{10cm} \text{By Lemma 6.16 of of \cite{markovmixing}}
	\end{align*}
 This implies that, 
	\begin{align*}
		||\mathcal{L}(\sigma^{2}_n)-\mathcal{L}(\sigma^{'2}_n)||_{TV} &\leq E\left[\sup_{x>0} \left\{ 1- \frac{x^{-\alpha-1}e^{-1/x}}{(1+\Delta/W_n)^{\alpha}x^{-\alpha-1}e^{-(1+\Delta/W_n)/x}} \right\} \right] \\
		&= E\left[\sup_{x>0} \left\{ 1- \frac{e^{\Delta/W_n/x}}{(1+\Delta/W_n)^{\alpha}} \right\} \right] \\
		&= E\left[1- \frac{1}{(1+\Delta/W_n)^{\alpha}}\right].
	\end{align*}
	Define $f(x)=\frac{1}{x^{\alpha}}, f'(x)=-\alpha\frac{1}{x^{\alpha+1}}$. By the mean value theorem, $f(1+\Delta/W_n) = f(1)-\frac{\Delta}{W_n}\frac{\alpha}{\xi^{\alpha+1}}, \xi \in (1, 1+\Delta/W_n)$. So,
	$f(1+\Delta/W_n)\geq 1-\alpha\frac{\Delta}{W_n}$.
	Now, 
	\begin{align*}
		&||\mathcal{L}(\sigma^{2}_n)-\mathcal{L}(\sigma^{'2}_n)||_{TV}\\ 
		&\leq E\left[f(1)-f(1+\frac{\Delta}{W_n})\right] \\
		&\leq E\left[1-(1-\alpha\frac{\Delta}{W_n})\right] = E\left[\alpha\frac{\Delta}{W_n}\right] \\
		&\leq E\left[\frac{k+\upsilon_0}{2}\Delta \frac{2}{\upsilon_0 c_0^2}\right] &\text{since $\alpha= \frac{k+\upsilon_0}{2}$ and $W_n\geq \frac{\upsilon_0 c_0^2}{2}$} \\
		&= \frac{k+\upsilon_0}{\upsilon_0 c_0^2} E[|W_n-W'_n|]
		\\
		&= \frac{k+\upsilon_0}{\upsilon_0 c_0^2} E[|W_n-W'_n|] E[G_n^{-1}] E[G_n^{-1}]^{-1}
		\\
		&= \frac{k+\upsilon_0}{\upsilon_0 c_0^2} E[|W_n / G_n - W'_n / G_n|] E[G_n^{-1}]^{-1} &\text{by independence}
		\\
		&\le \frac{(k+\upsilon_0)^2}{ 2 \upsilon_0 c_0^2} E[|W_n / G_n - W'_n / G_n|] \\
		&= \frac{(k+\upsilon_0)^2}{ 2 \upsilon_0 c_0^2} E[|\sigma^{2}_n - \sigma^{'2}_n|].
	\end{align*}
\end{proof}

\begin{proof}[Proof of Lemma \ref{lem:regKbound}]\label{pf:regKbound}

	Define
	$\alpha'=(k+\upsilon)/2-2$ and $\beta'=\upsilon_0 c^2_0/2$. Note that
	
	$$f_{\nu}(\beta,\sigma^2)=\frac{\beta'^{\alpha'}}{\Gamma(\alpha')}\frac{1}{(\sigma^{2})^{\alpha'+1}}e^{-\beta'/\sigma^{2}}\times(2\pi)^{-q/2}\det(\Sigma_{\beta})^{-1/2}e^{-\frac{1}{2}(\beta-\beta_0)^T\Sigma^{-1}_{\beta}(\beta-\beta_0)}$$
This implies that,
	\begin{align*}
		&\sup_{(\beta, \sigma^{2})\in \mathbb{R}\times \mathbb{R}_+}\frac{g(\beta, \sigma^{2})}{f_{\nu}(\beta, \sigma^{2})}\\ 
		&= \frac{\frac{1}{(\sigma^2)^{(k+\upsilon_0)/2+1}}\exp\left(-\frac{1}{2\sigma^2} (Y-X\beta)^T(Y-X\beta) -\frac{1}{2}(\beta-\beta_0)^T\Sigma_{\beta}^{-1} (\beta-\beta_0) -\frac{\upsilon_0 c_0^2}{2\sigma^2}\right)}
		{\frac{\beta'^{\alpha'}}{\Gamma(\alpha')}\frac{1}{(\sigma^{2})^{(k+\upsilon)/2-2+1}}e^{-\upsilon_0 c^2_0/2\sigma^{2}}(2\pi)^{-q/2}\det(\Sigma_{\beta})^{-1/2}e^{-\frac{1}{2}(\beta-\beta_0)^T\Sigma^{-1}_{\beta}(\beta-\beta_0)}} \\
		&= \frac{(\sigma^2)^{-2}\exp\left(-\frac{1}{2\sigma^2} (Y-X\beta)^T(Y-X\beta) \right)}
		{\frac{\beta'^{\alpha'}}{\Gamma(\alpha')}(2\pi)^{-q/2}\det(\Sigma_{\beta})^{-1/2}} 
	\end{align*}
		The numerator of the last equality is proportional to an inverse gamma density function with $\alpha^*=1$ and $\beta^*=\frac{1}{2} (Y-X\beta)^T(Y-X\beta)$. Thus, the mode occurs when $\sigma^{2*}=\frac{\beta^*}{\alpha^*+1}=\frac{(Y-X\beta)^T(Y-X\beta)}{4}$ and so,
		
		\begin{align*}
			&\sup_{(\beta, \sigma^{2})\in \mathbb{R}\times \mathbb{R}_+}\frac{g(\beta, \sigma^{2})}{f_{\nu}(\beta, \sigma^{2})}\\ 
			&\leq \frac{16}{((Y-X\beta)^T(Y-X\beta))^2} \frac{\exp\left(-\frac{4}{2(Y-X\beta)^T(Y-X\beta)} (Y-X\beta)^T(Y-X\beta) \right)}
			{\frac{\beta'^{\alpha'}}{\Gamma(\alpha')}(2\pi)^{-q/2}\det(\Sigma_{\beta})^{-1/2}} \\
			&= \frac{16}{((Y-X\beta)^T(Y-X\beta))^2} \frac{\Gamma(\alpha')(2\pi)^{q/2}\det(\Sigma_{\beta})^{1/2}}
			{\beta'^{\alpha'}e^{2}} \\
			&\leq \frac{16}{((Y-X\hat{\beta})^T(Y-X\hat{\beta}))^2} \frac{\Gamma(\alpha')(2\pi)^{q/2}\det(\Sigma_{\beta})^{1/2}}
			{\beta'^{\alpha'}e^{2}} \\
		\end{align*}
		
		Where $\hat{\beta}=(X^T X)^{-1}X^T Y$ is the maximum likelihood estimator.
	Next we find a lower bound on $\int_{B}g(\beta,\sigma^2) \mathrm{d}(\beta,\sigma^2)$
	
	\begin{align*}
		&\int_{\mathbb{R}^q\times \mathbb{R}_+}g(\beta,\sigma^2)\mathrm{d}(\beta,\sigma^2) \\
		&=\int_{\mathbb{R}^q\times \mathbb{R}_+}
		\frac{1}{(\sigma^2)^{(k+\upsilon_0)/2+1}}\\
		&\hspace{2cm}\exp\left(-\frac{1}{2\sigma^2} (Y-X\beta)^T(Y-X\beta) -\frac{1}{2}(\beta-\beta_0)^T\Sigma_{\beta}^{-1} (\beta-\beta_0) -\frac{\upsilon_0 c_0^2}{2\sigma^2}\right)\mathrm{d}(\beta,\sigma^2)\\
		&=\int_{\mathbb{R}^q}\left(\int_{\mathbb{R}_+}
		\frac{1}{(\sigma^2)^{(k+\upsilon_0)/2+1}}
		e^{-(\upsilon_0 c_0^2/2 +	 (Y-X\beta)^T(Y-X\beta)/2)/\sigma^2}\mathrm{d}(\sigma^2)\right)e^{-\frac{1}{2}(\beta-\beta_0)^T\Sigma_{\beta}^{-1} (\beta-\beta_0)}\mathrm{d}(\beta)\\
		&=\int_{\mathbb{R}^q}
		\frac{\Gamma((k+\upsilon_0)/2)}{(\upsilon_0 c_0^2/2 +	 (Y-X\beta)^T(Y-X\beta)/2)^{(k+\upsilon_0)/2}}
		e^{-\frac{1}{2}(\beta-\beta_0)^T\Sigma_{\beta}^{-1} (\beta-\beta_0)}\mathrm{d}(\beta)\\
		&=\Gamma((k+\upsilon_0)/2) \int_{\mathbb{R}^q}
		\frac{e^{-\frac{1}{2}(\beta-\beta_0)^T\Sigma_{\beta}^{-1} (\beta-\beta_0)}}
		{(\upsilon_0 c_0^2/2 +	 (Y-X\beta)^T(Y-X\beta)/2)^{(k+\upsilon_0)/2}}
		\mathrm{d}(\beta)\\
		&\leq\Gamma((k+\upsilon_0)/2) \int_{B}
		\frac{e^{-\frac{1}{2}(\beta-\beta_0)^T\Sigma_{\beta}^{-1} (\beta-\beta_0)}}
		{(\upsilon_0 c_0^2/2 +	 (Y-X\beta)^T(Y-X\beta)/2)^{(k+\upsilon_0)/2}}
		\mathrm{d}(\beta)\\
	\end{align*}
	where $B\subset \mathbb{R}^q\times \mathbb{R}_+$ is a bounded set.

Thus we get that 
	\begin{align*}
		K_{1}(\pi,\nu)&=\frac{1}{\int_{\mathbb{R}^q\times \mathbb{R}_+}g(\beta,\sigma^2)\mathrm{d}(\beta,\sigma^2)}\sup_{(\beta, \sigma^{2})\in \mathbb{R}\times \mathbb{R}_+}\frac{g(\beta, \sigma^{2})}{f_{\nu}(\beta, \sigma^{2})}\\
		&\leq\frac{1}{\Gamma((k+\upsilon_0)/2) \int_{B}
			\frac{e^{-\frac{1}{2}(\beta-\beta_0)^T\Sigma_{\beta}^{-1} (\beta-\beta_0)}}
			{(\upsilon_0 c_0^2/2 +	 (Y-X\beta)^T(Y-X\beta)/2)^{(k+\upsilon_0)/2}}
			\mathrm{d}(\beta)}\\
			&\hspace{2cm}\frac{16 \Gamma((k+\upsilon_0)/2-2)(2\pi)^{q/2}\det(\Sigma_{\beta})^{1/2}}{((Y-X\hat{\beta})^T(Y-X\hat{\beta}))^2 \left(\upsilon_0 c^2_0/2\right)^{(k+\upsilon_0)/2-2}e^{2}} \\
		&=\frac{1}{\int_{B}
			\frac{e^{-\frac{1}{2}(\beta-\beta_0)^T\Sigma_{\beta}^{-1} (\beta-\beta_0)}}
			{(\upsilon_0 c_0^2/2 +	 (Y-X\beta)^T(Y-X\beta)/2)^{(k+\upsilon_0)/2}}
			\mathrm{d}(\beta)}\\
			&\hspace{2cm}
			\frac{16 \Gamma((k+\upsilon_0)/2-2)(2\pi)^{q/2}\det(\Sigma_{\beta})^{1/2}}
			{\Gamma((k+\upsilon_0)/2)((Y-X\hat{\beta})^T(Y-X\hat{\beta}))^2 \left(\upsilon_0 c^2_0/2\right)^{(k+\upsilon_0)/2-2}e^{2}} 
	\end{align*}
So by Equation \ref{eqn:kupper} combined with Theorem \ref{thm:summary}, for $r\geq p$
\begin{align*}
	W_{d,p}(\mathcal{L}(X_n),\pi)
	\leq \left(\frac{C}{L}E[d(X_n,Y_n)^r]\right)^{1/r}
\end{align*}
Where $C=\frac{16 \Gamma((k+\upsilon_0)/2-2)(2\pi)^{q/2}\det(\Sigma_{\beta})^{1/2}}
{\Gamma((k+\upsilon_0)/2)((Y-X\hat{\beta})^T(Y-X\hat{\beta}))^2 \left(\upsilon_0 c^2_0/2\right)^{(k+\upsilon_0)/2-2}e^{2}}$ \\ and $L=\int_{B}
\frac{e^{-\frac{1}{2}(\beta-\beta_0)^T\Sigma_{\beta}^{-1} (\beta-\beta_0)}}
{(\upsilon_0 c_0^2/2 +	 (Y-X\beta)^T(Y-X\beta)/2)^{(k+\upsilon_0)/2}}
\mathrm{d}(\beta)$
\end{proof}

	\begin{remark}\label{rem:linreg}
	For Real Data Example \ref{numex:linreg}, we want to prove that $E[||(\beta_n, \sigma^2_n)-(\beta'_n, \sigma^{'2}_n)||_1]<\infty$ for $n\in\mathbb{N}$.

	Since $\sigma^2_0=100$ and $\sigma^{'2}_0 \sim IG((k+\upsilon_0)/2, \upsilon_0 c_0^2/2)$, $E[|\sigma^{2}_0|]<\infty$ and $E[|\sigma^{'2}_0|]<\infty$.
	
	Suppose that $E[|\sigma^2_n|]<\infty$. By Equation \ref{eqn:sigma_linereg}, $\sigma^{2}_n=k_{Z_n,IG_n}(\sigma^{2}_{n-1})$, so \\ $E[|\sigma^{2}_n|]= \allowbreak E[E[k_{Z_n,IG_n}(\sigma^{2}_{n-1})|\sigma^{2}_{n-1}]]<\infty$.
	
	Further, since $\beta_n \sim N_q(\tilde{\beta}_{\sigma^2_{n}}, V_{\beta, \sigma^2_{n}})$ and $\beta'_n \sim N_q(\tilde{\beta}_{\sigma^{'2}_{n}}, V_{\beta, \sigma^{'2}_{n}})$, $E[|\beta_n|]<\infty$ and $E[|\beta'_n|]<\infty$.

	Thus, we have established that $E[||(\sigma^{2}_n,\beta_n)||_1]<\infty$ and $E[||(\sigma^{'2}_n,\beta'_n)||_1]<\infty$. By the triangle inequality, $E[||(\sigma^{2}_n,\beta_n)-(\sigma^{'2}_n,\beta'_n)||_1]<\infty$.
	
%
\end{remark}

\bibliographystyle{plain} 
\bibliography{bibliography}
\end{document}